\documentclass{article}
\usepackage[utf8]{inputenc}

\usepackage{pdfpages}
\usepackage[margin=1in]{geometry} 
\usepackage{amsmath,amsthm,amssymb}
\usepackage[utf8]{inputenc}
\usepackage{graphicx}
\usepackage{subfigure}

\newcommand{\R}{\mathbb{R}}
\usepackage[hidelinks]{hyperref}
\usepackage{float}
\usepackage{mathtools,cases}
\usepackage{bbm}
\usepackage{calrsfs}
\usepackage{breqn}
\usepackage{multirow}
\usepackage{xcolor}
\usepackage{comment}
\usepackage{soul}
\usepackage{bbold}

\usepackage[numbers]{natbib}

\newtheorem{theorem}{Theorem}[section]
\newtheorem{lemma}[theorem]{Lemma}

\newtheorem{corollary}[theorem]{Corollary}

\newtheorem{remark}[theorem]{Remark}

\newtheorem{definition}[theorem]{Definition}

\numberwithin{equation}{section}

\newcommand{\bl}{\color{blue}}

\usepackage{enumitem}

\usepackage{mathtools}
\usepackage{authblk}
\mathtoolsset{showonlyrefs}
\title{Reconciling rough volatility with jumps}

\author[1]{Eduardo Abi Jaber\thanks{eduardo.abi-jaber@polytechnique.edu. The first author is grateful for the financial support from the Chaires FiME-FDD, Financial Risks, Deep Finance \& Statistics and Machine Learning and systematic methods in finance at Ecole Polytechnique.}}
\author[2]{Nathan De Carvalho\thanks{nathan.decarvalho@engie.com. The second author is grateful for the financial support provided by Engie Global Markets. We would like to thank
Ryan McCrickerd and Andreas S{\o}jmark for fruitful discussions.}}
\affil[1]{Ecole Polytechnique, CMAP}
\affil[2]{Engie Global Markets $\&$ Université Paris Cité and Sorbonne Université, LPSM}

\begin{document}

\maketitle

\begin{abstract}
We reconcile rough volatility models and jump models using a class of \textit{reversionary Heston models} with fast mean reversions and large vol-of-vols.
Starting from hyper-rough Heston models with a Hurst index $H\in(-1/2,1/2)$, we derive a Markovian approximating class of one dimensional  \textit{reversionary Heston-type models}. Such proxies encode a trade-off between an exploding vol-of-vol and a fast mean-reversion speed controlled by a reversionary time-scale $\epsilon>0$ and an unconstrained parameter $H \in \mathbb{R}$. Sending  $\epsilon$ to $0$ yields convergence  of the \textit{reversionary Heston model} towards different explicit asymptotic regimes based on the value of the parameter $H$. In particular, for $H \le -1/2$, \textit{the reversionary Heston model} converges to a class of  Lévy jump processes of Normal Inverse Gaussian type.  Numerical illustrations  show that the \textit{reversionary Heston model} is capable of generating at-the-money skews similar to the  ones generated by rough, hyper-rough and jump models.
\end{abstract} 

\begin{description}
\item[Mathematics Subject Classification (2010):] 91G20, 60G22, 60G51 
\item[JEL Classification:] G13, C63, G10.   
    \item[Keywords:]  Stochastic volatility, Heston model, Normal Inverse Gaussian, rough Heston model,  Riccati equations
\end{description}

\newpage
\tableofcontents

\section*{Introduction}


\noindent Since the 1987 financial crash, financial option markets have exhibited a notable implied volatility skew, especially for short-term maturities. This skew reflects the market's expectation of significant price movements on very short time scales in the underlying asset, which poses a challenge to traditional continuous models based on standard Brownian motion. To address this issue, the literature has developed several classes of models that capture the skewness in implied volatilities. Three prominent approaches  are:


\begin{itemize}
    \item conventional  one-factor stochastic volatility models boosted with large mean-reversion speed and vol-of-vol. This class of models have been  justified by  several empirical studies that have identified the presence of very fast mean-reversion in the S$\&$P volatility time series \cite{alizadeh2002range,chernov2003alternative, fouque2003multiscale,fouque2003short} and by the fact that they are able to correct conventional models  to reproduce the behavior of the at-the-money skew for short maturities  \cite{mechkov2015fast};
    
    \item jump diffusion models, especially the class of affine jump-diffusions for which valuation problems become (semi-)explicit using Fourier inversion techniques, see \cite{duffie_affine_jump_transforms_2003}. Such class of models incorporates  occasional and large  jumps to explain the   skew  observed implicitly on option markets, see \cite{tankov2003financial}, and \cite{bakshi_alternative_option_pricing_1997} for an empirical analysis of the impact of adding jumps to stochastic volatility diffusion on the implied volatility surface;
    
    \item rough volatility models, where the volatility process is driven by variants of the Riemann-Liouville fractional Brownian motion
    \begin{equation} \label{eq:RLfBM}
        W_t^{H} = \frac{1}{\Gamma \left( H+1/2 \right)}\int_0^t \left( t-s\right)^{H-1/2}dW_s, \quad t \geq 0,
    \end{equation}
    with $W$  a standard Brownian motion and $H \in \left(0,1/2\right)$  the Hurst index. Such models are  able to reproduce  the roughness of the spot variance's trajectories measured empirically \cite{gatheral2018volatility, bennedsen2016decoupling} together with  explosive behaviors of the at-the-money skew \cite{alos2006short, bayer2016pricing, el2019characteristic,  fukasawa2011asymptotic, abi2022characteristic}.
\end{itemize}

\noindent So far, in the mathematical finance community, jump diffusion models and rough volatility models have  often been treated as distinct approaches, and, in some cases, they have even been opposed to each other, see for instance  \cite[Section 5.3.1]{bayer2016pricing}. However, on the one side,  connections between rough volatility models and fast mean-reverting factors have  been established in \cite{abi2019lifting, abi2019markovian, abi2019multifactor}. On the other side, jump models have been related to fast regimes stochastic volatility models in \cite{mechkov2015fast,McCrickerd_2021}. In parallel, from the empirical point of view,  it can be very challenging for the human eye and for statistical estimators to distinguish between roughness, fast mean-reversions and jump-like behavior, as shown in \cite{abi2019multifactor, cont2022rough, garcin2022long}. \\

\noindent The above suggests that rough volatility and jump models may not be that different after all. Our main motivation is to establish for the first time in the literature  a connection between rough volatility and jump models through conventional volatility models with fast mean-reverting regimes. \\
 
 \noindent We aim to reconcile these two classes of models through the use of the celebrated conventional Heston model \cite{heston1993} but with a parametric specification which encodes a trade-off between a fast mean-reversion and a large vol-of-vol.  We define the \textit{reversionary Heston models}  as follow:
    \begin{align} \label{eq:reversionaryHestonInstantaneousSpot}
      & dS_t^\epsilon = S_t^\epsilon \sqrt{V_t^{\epsilon}} \left( \rho dW_t + \sqrt{1-\rho^2} dW_t^\perp \right), \quad S_0^\epsilon=S_0,\\ \label{eq:reversionaryHestonInstantaneousVol}
      & dV^\epsilon_t =  \left( \epsilon^{H- \frac{1}{2}} \theta - \epsilon^{-1} \left( V_t^{\epsilon}-V_0 \right) \right) dt + \epsilon^{H-\frac{1}{2}} \xi \sqrt{V_t^{\epsilon}}dW_t, \quad V^\epsilon_0=V_0,
    \end{align}
    where $\left( W, W^\perp \right)$ is a two-dimensional Brownian motion, $\theta \ge 0$, $S_0, \xi, V_0 >0$, $\rho \in [-1,1]$. The two crucial parameters here are the \textit{reversionary time-scale} $\epsilon>0$ and $H \in \mathbb{R}$. Such parametrizations nest as special cases the fast regimes extensively studied by \citet{fouque2003multiscale, feng2010shortMaturityAsymptotics}, see also \cite[Section 3.6]{fouque2000derivatives}, which correspond to the case $H=0$; 
    and also the regimes studied in \cite{mechkov2015fast,McCrickerd_2021} for the case $H = -1/2$. Letting the parameter $H\in \mathbb R$ free in \eqref{eq:reversionaryHestonInstantaneousVol}  introduces more flexibility in practice    and leads to better fits with stable calibrated parameters across time as recently shown in \cite{abijaber2022joint, abijaber2022quintic}. In theory, it allows for a better understanding of the impact of the scaling in $H$ on the limiting behavior of the model as $\epsilon\to 0$  as highlighted in this paper.
    \\

\noindent In a nutshell, we show that:

\begin{enumerate}[label=(\roman*)]
    \item 
        {for $H>-1/2$, the \textit{reversionary Heston model}  can be constructed as a proxy of rough and hyper-rough Heston models  where $H \in (-1/2,1/2]$ plays  a similar role to that of}  the Hurst index,
        \item 
        for $H\leq -1/2$, as $\epsilon \to 0$,
        the \textit{reversionary Heston model}  converges towards Lévy jump processes of Normal Inverse Gaussian type with distinct regimes for $H=-1/2$ and $H<-1/2$ respectively,  
        \item 
        the \textit{reversionary Heston model} is capable of generating implied volatility surfaces and  at-the-money skews similar to the ones generated by rough, hyper-rough and jump models, and comes arbitrarily close to the at-the-money skew scaling as $\tau^{-0.5}$ for small $\tau$, contrary to widespread understanding. 
    \end{enumerate}

\noindent Our results allow for a  reconciliation between rough and jump models as they  suggest that jump models and (hyper-)rough volatility models are complementary, and do not overlap. For $H>-1/2$, the \textit{reversionary Heston model} can be interpreted as an engineering proxy of rough and hyper-rough volatility models, while for $H\leq -1/2$, it corresponds to  an approximation of jump models for small enough $\epsilon$. Asymptotically, jump models actually start at $H=-1/2$ (and below) in the Reversionary Heston model, the very first value of the Hurst index for which hyper-rough volatility models can no-longer be defined. \\

\noindent More precisely,  our argument is structured as follows.
First, in Section~\ref{hyperRoughToReversionary}, we show how the reversionary Heston model \eqref{eq:reversionaryHestonInstantaneousSpot}-\eqref{eq:reversionaryHestonInstantaneousVol} can be obtained as a Markovian and semimartingale proxy of rough and hyper-rough Heston models \cite{el2019characteristic,jusselin2020no} with Hurst index $H \in (-1/2,1/2)$. This is achieved using the resolvent of the first kind of  the shifted fractional kernel. 

\noindent Second, in Section~\ref{S:jointchar}, we derive the joint conditional characteristic functional of the log-price $\log S^{\epsilon}$ and the integrated variance $\bar{V}^\epsilon := \int_0^{\cdot} V^\epsilon_s ds$ in the model  \eqref{eq:reversionaryHestonInstantaneousSpot}--\eqref{eq:reversionaryHestonInstantaneousVol} in terms of a solution to a system of time-dependent Riccati ordinary differential equations; see Theorem~\ref{T:charfuneps}. Compared to the literature, we  provide a  novel and  concise proof for the existence and uniqueness of a global solution to such Riccati equations using the variation of constant formulas.

\noindent Finally, in Section~\ref{S:convjumps}, we establish the convergence of the log-price and the integrated variance $(\log S^{\epsilon},\bar V^{\epsilon})$ in  the reversionary Heston model \eqref{eq:reversionaryHestonInstantaneousSpot}-\eqref{eq:reversionaryHestonInstantaneousVol} towards a Lévy jump process $(X,Y)$, as $\epsilon$ goes to $0$. More precisely, we show that the limit $(X, Y)$ belongs to the class of Normal Inverse Gaussian - Inverse Gaussian (NIG-IG) processes which we construct from its Lévy exponent and we connect such class to first hitting-time representations in the same spirit of \citet{barndorff_nig_1997}. Our main 
Theorem~\ref{T:charfunlimit} provides the convergence of the finite-dimensional distributions of the joint process $(\log S^{\epsilon}, \bar V^{\epsilon})$  through the  study of the limiting behavior  of the Riccati equations and hence the  characteristic functional given in Theorem~\ref{T:charfuneps}. Interestingly, the limiting behavior disentangles three different asymptotic regimes based on the values of $H$. The convergence of the integrated variance process  is even strengthened  to a functional weak convergence on the Skorokhod space   of càdlàg paths on $[0,T]$ endowed with the $M_1$ topology. We stress that the usual $J_1$ topology is not useful here, since jump processes cannot be obtained as limits of continuous processes in the $J_1$ topology. \\

\noindent \textbf{Related Literature.} 
 Convergence of the reversionary Heston models towards jump processes: our results clarify and extend the results of \cite{mechkov2015fast,McCrickerd_2021}, derived for the case $H=-1/2$, that establish and make clear the precise limiting connection between the Heston log-price process and the normal inverse-Gaussian (NIG) process of \cite{barndorff_nig_1997}. Connections between the long time behavior of the Heston log-price process and NIG distribution were first exposed  in \cite{forde2011large, keller2011moment} and were the main motivations behind the work of \citet{mechkov2015fast}.\\
\noindent Relevance of fast regimes in practice: {the pricing of options near maturity is challenging because of the very steep slope of smiles observed on the market and \citet{fouque2003multiscale} showed that stochastic volatility should embed both a fast regime Ornstein-Uhlenbeck factor (see Remark \ref{rk:fBMproxy} below) from which approximations of option prices can be derived using a singular perturbation expansion, and a slowly varying factor to be able to match options with long maturities. On the other hand, \citet{feng2010shortMaturityAsymptotics} considers a Heston model with a fast mean-reverting volatility and uses large deviation theory techniques to derive an approximation price for out-of-the-money vanilla options when the maturity is small, but large compared to the characteristic time-scale of the stochastic volatility factor. More recently an Ornstein-Uhlenbeck process with the same parametrization as in \eqref{eq:reversionaryHestonInstantaneousVol} has been used to construct  the Quintic stochastic volatility model \cite{abijaber2022quintic} to achieve remarkable joint fits of SPX and VIX implied volatilities, outperforming its rough and path-dependent counterparts as shown empirically in \cite{abijaber2022joint}.}\\

\noindent \textbf{Notations.} For $p\geq 1$, we denote by $L^p_{loc}$ the space of measurable functions $f:\R_+\to \R$ such that $\int_0^T |f(s)|^p ds<\infty$, for all $T>0$. We will denote by $\sqrt{x}$ the principal square root of $x \in \mathbb{C}$, i.e.~its argument lies within $(-\pi/2,\pi/2]$. 

\section{From rough Heston to reversionary Heston} \label{hyperRoughToReversionary}

In this section, we show how reversionary Heston models \eqref{eq:reversionaryHestonInstantaneousSpot}-\eqref{eq:reversionaryHestonInstantaneousVol} can be seen as proxies of rough and hyper-rough Heston models whenever $H>-1/2$. 

\subsection{Rough and hyper-rough Heston}
Let $\left( W, W^\perp \right)$ be a two-dimensional Brownian motion defined on a filtered probability space $\left( \Omega, \mathcal F, \left( \mathcal{F}_t \right)_{t \geq 0}, \mathbb Q \right)$ which satisfies the usual conditions, where $\mathbb Q$ is the risk-neutral probability. Set $B := \rho W + \sqrt{1-\rho^2} W^\perp$ with $\rho \in \left[ -1, 1 \right]$.  We take as a starting point a stochastic volatility model for an underlying asset $P$ in terms of a time-changed Brownian motion:
\begin{equation} \label{HyperRoughSpot}
    dP_t = P_t dB_{\bar{U}_t},  \quad P_0>0,
\end{equation}
 for some non-decreasing continuous process $\bar{U}$. If $\bar{U}_t = \int_0^t U_s ds$, then $U$ would correspond to the spot variance and $\bar{U}$ plays the role of the integrated variance. The hyper-rough Volterra Heston model introduced in \cite{jusselin2020no} and studied further in \cite[Section 7]{abi2021weak} assumes that the dynamics of the integrated variance is of the form
\begin{equation} \label{HyperRoughIntegratedVariance}
    \bar{U}_t = \bar{G}_0(t) + \xi \int_0^t K_H(t-s)  W_{\bar{U}_s} ds,
\end{equation}
for a suitable continuous function $\bar{G}_0$, and $\xi > 0$, and $K_H$ is the fractional kernel
\begin{equation}\label{eq:fractionalkernel}
    K_H(t)  =  t^{H-1/2}, \quad t >0,
\end{equation}
 for $H \in (-1/2,1/2]$. The lower bound $H>-1/2$ ensures the $L^1_{loc}$ integrability of the kernel $K_H$ so that the stochastic convolution appearing in \eqref{HyperRoughIntegratedVariance} is well-defined. Furthermore if the kernel happens to be in $L^2_{loc}$, the following lemma ensures the existence of a spot variance process.

\begin{lemma}[Existence of spot variance]\label{existenceSpotVariance}
    Let $K \in L^2_{loc}$ and $g_0 \in L^1_{loc}$. Assume there exists a non-decreasing adapted process $\bar{U}$ and a Brownian motion $W$ such that
    \begin{equation} \label{VolterraIntegratedVariance}
        \bar{U}_t  = \int_0^t g_0(s)ds + \int_0^t K(t-s) W_{\bar{U}_s} ds,
    \end{equation}
    with $\sup_{t\le T} \mathbb{E} \left[ \left| \bar{U}_t \right| \right] < \infty$, for all $T>0$. Then, $\bar{U}_t = \int_0^t U_s ds$, where $U$ is a non-negative weak solution to the following stochastic Volterra equation
    \begin{equation} \label{VolterraSpotVariance}
        U_t = g_0(t) + \int_0^t K(t-s) \sqrt{U_s}dW_s,  \quad \mathbb{Q} \otimes dt-a.e.
    \end{equation}
    Conversely, assume there exists a non-negative weak solution $U$ to the stochastic Volterra equation \eqref{VolterraSpotVariance} such that $\sup_{t\le T} \mathbb{E} \left[ U^2_t \right] < \infty$, for all $T>0$, then $\bar{U}$ solves \eqref{VolterraIntegratedVariance}.
\end{lemma}

\begin{proof}
   This is obtained by an application of stochastic Fubini's theorem, see \cite[Lemma 2.1]{abi2021weak}.
\end{proof}

\noindent Going back to the fractional case, if we restrict $H$  in $\left( 0, 1/2 \right]$, then we have $K_H \in L^2_{loc}$. For $\bar{G}_0(t) := \int_0^t g_0(s) ds$, a direct application of Lemma \ref{existenceSpotVariance} yields that the model \eqref{HyperRoughSpot}-\eqref{HyperRoughIntegratedVariance} is  equivalent to the rough Heston model of \citet{el2019characteristic} written in spot-variance form
\begin{align}
    & dP_t = P_t\sqrt{U_t}dB_t, \\
    & U_t = g_0(t) + \int_0^t K_H(t-s) \xi \sqrt{U_s} dW_s,
\end{align}
for some initial input curve $g_0: \mathbb{R}_+ \rightarrow \mathbb{R}$ ensuring the non-negativity of $V$. Two notable specifications of such admissible input curves  are given by \cite[Example 2.2]{abi2019markovian} and read 
\begin{equation}
  \begin{cases}
     g_0 \text{ continuous and non-decreasing with } g_0 \ge 0,\\
     \text{or}\\
     g_0(t) = U_0 + \theta \int_0^t K_H(s)ds, \text{ for some } U_0, \theta \ge 0.
 \end{cases}
\end{equation}
  \\
Moreover, for $H\in (0,1/2]$ the sample paths of the spot variance $U$ are locally Hölder continuous of any order strictly less than $H$, and consequently rougher than those of the standard Brownian motion, which corresponds to the  case $H=1/2$, justifying the denomination  `rough model'.   The hyper-rough appellation corresponds to the case $H\in (-1/2,0]$ for which the process $\bar U$ is continuous but no longer absolutely continuous. Indeed, in this case, one can show that the trajectories of $U$ are nowhere differentiable, see \cite[Proposition 4.6]{jusselin2020no}.\\

\noindent A key advantage of rough and hyper-rough Heston models is the semi-explicit knowledge of the characteristic function of the log-price modulo a deterministic Riccati Volterra convolution equation, as they belong to the class of Affine Volterra processes \cite{abi2019affine,abi2021weak}. More precisely,   
for any $u=(u_1,u_2) \in \mathbb C^2$  satisfying  $$\Re(u_1)=0,\quad \Re(u_2) \leq 0,$$ 
	the joint Fourier--Laplace transform of $(\log P, \bar U)$  is given by
	$$\mathbb E\left[ \exp\left( u_1\log P_T + u_2 \bar U_T  \right) \right] = \exp\left(  u_1 \log P_0 +   \int_0^T  R(\psi_H(T-s))d\bar G_0(s) \right), $$
	for all $T \geq 0$, where $\psi_H$ is the continuous solution to the following fractional Riccati--Volterra equation
\begin{align}\label{eq:VolterraRiccati}
	\psi_H(t)&=\int_0^t K_H(t-s)  R(\psi_H(s))ds, \quad t \geq 0, \nonumber \\
	R(x)&=\frac 12 (u_1^2 -u_1) +u_2 +  \rho \xi  u_1 x + \frac {\nu^2} 2 x^2,
	\end{align}
see \cite[Section 7]{abi2021weak}.
This allows for fast pricing and calibration via Fourier inversion techniques. Compared to the conventional Heston model where the characteristic function is explicit, the solution to the Riccati Volterra equation is not explicitly known and must be approximated numerically. Several numerical schemes have been proposed including the Adams scheme in \cite{el2019characteristic}, the multi-factorial approximation of rough volatility models in \cite{abi2019multifactor} or the hybrid scheme based on fractional power series expansion in \cite{pages2020roughNotSoTough}.

\subsection{Deriving  reversionary Heston as a proxy: $\epsilon$-shifting  the singularity} \label{ss:proxy_derivation}
In both regimes,  rough and hyper-rough, with the exception of $H=\frac 12 $, the model is non-Markovian, non-semimartingale with singular kernels. From a practitioner standpoint, it  is therefore natural to look  for  Markovian approximations by suitable smoothing  of the singularity of the fractional kernel \eqref{eq:fractionalkernel} sitting at the origin.
In this section, we show how we can build a Markovian semi-martingale proxy of hyper-rough models. This is achieved using a two-step procedure. \\

\textbf{First step: recover semimartingality} by smoothing out the singularity of the fractional kernel $K_H$. 
We fix $\epsilon>0$, and  we  consider the shifted fractional kernel 
\begin{equation}
    K_{H, \epsilon}(t) := \left( t+ \epsilon \right)^{H-\frac{1}{2}}, \quad t> 0, 
\end{equation}
and the corresponding `integrated variance' $\bar U^{\epsilon}$ given by 
\begin{equation}     \bar{U}^{\epsilon}_t = \int_0^t g_0^{\epsilon}(s) ds + \xi \int_0^t K_{H, \epsilon}(t-s)  W_{\bar{U}^{\epsilon}_s} ds,
\end{equation}
with   
$$g_0^{\epsilon}(t) = U_0 + \theta \int_0^t K_{H,\epsilon}(s)ds. $$
Note that now $K_{H, \epsilon}$ is in $L^2_{loc}$ for any value of $H$, so that an application of Lemma~\ref{existenceSpotVariance} yields that $\bar U^{\epsilon} = \int_0^{\cdot} U^{ \epsilon}_s ds$ where the spot variance $U^{\epsilon}$ solves the equation 
\begin{align}
     dP^{\epsilon}_t &= P^{\epsilon}_t \sqrt{U^{\epsilon}_t}dB_t \\
     U^{\epsilon}_t &= U_0 + \int_0^t K_{H,\epsilon}(t-s) \theta ds  + \int_0^t K_{H,\epsilon}(t-s) \xi \sqrt{U^{\epsilon}_s} dW_s.
\end{align}
Moreover, since $K_{H,\epsilon}$ is continuously differentiable on $[0,T]$, denoting by $K'_{H,\epsilon}$ its derivative, we get that $U^{\epsilon}$ is a semimartingale with the following dynamics 
\begin{align}\label{eq:Ueps}
    dU_t^{\epsilon} =   \left(K_{H,\epsilon} (0)\theta + \int_0^t K'_{H,\epsilon}(t-s) dZ_s^{\epsilon} \right)dt + K_{H,\epsilon} (0) \xi \sqrt{U_t^{\epsilon}} dW_t,
\end{align}
with
\begin{align}
dZ^{\epsilon}_t  = \theta dt + \xi \sqrt{U_t^{\epsilon}} dW_t. 
\end{align}

\textbf{Second step: recover a Markovian proxy thanks to the resolvent of the first kind.}  The only non-Markovian term in \eqref{eq:Ueps} is the term $\int_0^t K'_{H,\epsilon}(t-s) dZ_s^{\epsilon}$ appearing in the drift. Using the resolvent of the first kind of $K_{H,\epsilon}$ we will re-express this term in terms of a functional of the past of the process $U^{\epsilon}$.  For a kernel $K$, a {resolvent of the first kind} is a  measure $L$ on $\R_+$ of locally bounded variation such that
\begin{equation} \label{res_L}
\int_{[0,t]} K(t-s)L(ds) = 1, \quad t\geq 0,
\end{equation}
see \cite[Definition~5.5.1]{GripenbergVolterraIntegral}.  A resolvent of the first kind does not always exist.  We will make use of the notations $(f*g)(t)=\int_{0}^tf(t-s)g(s)ds$ and $(f*L)(t)=\int_{[0,t]}f(t-s)L(ds)$.

\begin{lemma}
    Fix $\epsilon>0$ and $ H \in (-1/2,1/2)$. The kernel $K_{H, \epsilon}$ admits a resolvent of the first kind $L_{\epsilon}$  of the form
\begin{equation} \label{ResolventFirstKind}
    L_{\epsilon} (dt)= \frac{\delta_0 (dt)}{K_{H,\epsilon}(0)} + \ell_{\epsilon}(t)dt,
\end{equation}
with $\ell_{\epsilon}$ a locally integrable function.
Moreover,  the function $( K_{H, \epsilon}' * L_{\epsilon})$ is continuously differentiable and it holds, for all $t\geq 0$, that
\begin{align}\label{eq:K'Z}
    \int_0^t K'_{H,\epsilon}(t-s) dZ_s^{\epsilon}  =  - \left( \frac{1}{2}-H \right) \epsilon^{-1} \left( U_t^{\epsilon}-U_0\right)  + \int_0^t \left( K_{H, \epsilon}' * L_{\epsilon} \right)'(t-s)   \left( U_s^{\epsilon} - U_0\right) ds.
\end{align}
\end{lemma}

\begin{proof}
First, the existence of the resolvent is justified as follows. Given $ H \in (-\frac{1}{2},\frac{1}{2})$, $K_{H, \epsilon}$ is a positive completely monotone kernel\footnote{Recall that a function $f$ is completely monotone if it is infinitely differentiable on $(0,\infty)$ such that $(-1)^n f^{(n)} \geq 0$, for all $n \geq 0$.} on $[0,T]$ so that an  application of \cite[Theorem 5.5.4]{GripenbergVolterraIntegral} yields the existence of a resolvent of the first kind in the form \eqref{ResolventFirstKind} with $\ell_{\epsilon}$ a completely monotone function. Convolving \eqref{ResolventFirstKind} with $K'_{H,\epsilon}$ one obtains that 
$$ (K'_{H,\epsilon} * L_{\epsilon})(t)  = \frac{K'_{H,\epsilon}(0)}{K_{H,\epsilon}(0)} +  (K'_{H,\epsilon} * \ell_{\epsilon})(t).$$
Since $K_{H,\epsilon}$ is twice continuously differentiable on $[0,T]$ and $\ell_{\epsilon}$ is integrable, it follows that $(K'_{H,\epsilon} * \ell_{\epsilon})$ is continuously differentiable and so is $(K'_{H,\epsilon} * L_{\epsilon})$. Noting that $(1*f)(.) := \int_0^. f(s) ds$ and applying the fundamental theorem of calculus on the function $(K_{H, \epsilon}' * L_{\epsilon})$ yields
\begin{equation} \notag
    \left( K_{H, \epsilon}' * L_{\epsilon} \right) = \left( K_{H, \epsilon}' * L_{\epsilon} \right)(0)+ 1 * \left( K_{H, \epsilon}' * L_{\epsilon} \right)',
\end{equation}
convolving on the right-hand side by $K_{H, \epsilon}$ combined with the associativity of the convolution operation and the fact that $(L_{\epsilon}*K_{H, \epsilon})=1$ yields:
\begin{equation} \notag
K_{H, \epsilon}' * 1 = K_{H, \epsilon}' *( L_{\epsilon}*K_{H, \epsilon}) = (K_{H, \epsilon}' * L_{\epsilon})*K_{H, \epsilon} = \left( K_{H, \epsilon}' * L_{\epsilon} \right)(0) (1*K_{H, \epsilon}) + 1 * \left( K_{H, \epsilon}' * L_{\epsilon} \right)'*K_{H, \epsilon}.
\end{equation}
And thus, we obtain almost everywhere with regards to the Lebesgue measure that:
\begin{equation}
K_{H, \epsilon}' = \left( K_{H, \epsilon}' * L_{\epsilon} \right)(0) K_{H, \epsilon} + (K_{H, \epsilon}' * L_{\epsilon})' * K_{H, \epsilon}.
\end{equation}
In addition, using \eqref{ResolventFirstKind}, we notice that
\begin{equation} \notag
    \left( K_{H, \epsilon}' * L_{\epsilon} \right)(0)= \frac{K_{H, \epsilon}'(0)}{K_{H, \epsilon}(0)} = \left( H-\frac{1}{2} \right) \epsilon^{-1}.
\end{equation}
Combining the above, we obtain that 
\begin{align}
    \int_0^t K'_{H,\epsilon}(t-s) dZ_s^{\epsilon}  =  \left( H-\frac{1}{2} \right) \epsilon^{-1} \int_0^t K_{H,\epsilon}(t-s) dZ_s^{\epsilon}    + \left((K_{H, \epsilon}' * L_{\epsilon})' * (K_{H, \epsilon}*dZ_t^{\epsilon})\right)_t, 
\end{align}
which yields \eqref{eq:K'Z}, after recalling that $U^{\epsilon} - U_0 = \int_0^{\cdot} K_{H,\epsilon}(\cdot-s) dZ_s^{\epsilon}$.
\end{proof}

\noindent With the help of the resolvent of the first kind, we were able to recover in the first term of \eqref{eq:K'Z} the first order mean-reversion scale  of the fractional kernel. The second term in \eqref{eq:K'Z} depends on the whole past trajectory of $U^{\epsilon}$.\\

\noindent We can now derive our Markovian proxy of the hyper-rough Heston model as follows: plugging the expression \eqref{eq:K'Z} in the drift of \eqref{eq:Ueps}, recalling that $K_{H,\epsilon}(0)=\epsilon^{H-1/2}$ and dropping the non-Markovian term $\left( \left( K_{H, \epsilon}' * L_{\epsilon} \right)_{.}' * \left( U^{\epsilon}_{\cdot} - U_0^{\epsilon} \right) \right)_t$, we arrive to the Markovian process:
\begin{equation} \label{ProxyRoughHeston}
    d \tilde V_t^{\epsilon} = \left(  \epsilon^{H-\frac{1}{2}} \theta  -  \left( \frac{1}{2}-H \right) \epsilon^{-1} \left( \tilde V_t^{\epsilon} - U_0 \right) \right) dt + \epsilon^{H-\frac{1}{2}} \xi \sqrt{\tilde V_t^{\epsilon}} dW_t, \quad \tilde V_0 = U_0.
\end{equation}

\noindent Finally, {re-scaling the mean-reversion speed from $\left( \frac{1}{2}-H \right) \epsilon^{-1}$ to $\epsilon^{-1}$ }leads to our reversionary Heston model  \eqref{eq:reversionaryHestonInstantaneousSpot}--\eqref{eq:reversionaryHestonInstantaneousVol} {where the parameter $H$ becomes unconstrained}.  In the following section, we illustrate numerically the fact that such reversionary Heston model can be seen as a proxy of rough and hyper-rough Heston models. 

\begin{remark}\label{rk:fBMproxy}
    Such proxy approximation can directly be applied to the Riemann-Liouville fractional Brownian motion defined in \eqref{eq:RLfBM} to get the proxy:
    \begin{equation}
        d W_t^{H,\epsilon} = -  \left( \frac{1}{2}-H \right) \epsilon^{-1} W_t^{H,\epsilon} dt + \epsilon^{H-\frac{1}{2}} dW_t, \quad W_0^{H,\epsilon} = 0,
    \end{equation}
    whose solution is explicitly given by
    \begin{equation}
        W_t^{H,\epsilon} = \epsilon^{H-\frac{1}{2}} \int_0^t e^{- \left( \frac{1}{2}-H \right) \epsilon^{-1} (t-s)} dW_s,
    \end{equation}
    \noindent which is a mean-reverting Ornstein-Uhlenbeck process as long as $H<\frac{1}{2}$, while the value $H=\frac{1}{2}$ yields back the standard Brownian motion.\\ First, such Ornstein-Uhlenbeck process has been recently used to construct  the Quintic stochastic volatility model \cite{abijaber2022quintic} to achieve remarkable joint fits of SPX and VIX implied volatilies, outperforming even its rough and path-dependent counterparts as shown empirically in \cite{abijaber2022joint}. \\
    {\noindent Furthermore, notice that the case $H=0$ degenerates into the fast scale volatility factor from \citet{fouque2003multiscale}, with $m=0$, $\nu=1$ and their time-scale is twice the reversionary time-scale $\epsilon$, and whose auto-correlation under the invariant distribution is given by
    \begin{equation*}
        \mathbb E \left[ W_t^{0,\epsilon} W_s^{0,\epsilon} \right] = e^{-\frac{\left|t-s\right|}{2\epsilon}}.
    \end{equation*}
    Consequently, the reversionary time-scale $\epsilon$ sets the speed of decay of the auto-correlation function of $W^{H,\epsilon}$.}
\end{remark}

\subsection{Numerical illustration}

\noindent In Section~\ref{ss:proxy_derivation}, we derived the \textit{Reversionary Heston} model \eqref{eq:reversionaryHestonInstantaneousSpot}--\eqref{eq:reversionaryHestonInstantaneousVol} as a semi-martingale and Markovian proxy of rough and hyper-rough Heston models such that the parameter $H$, stemming naturally from the Hurst index, is unconstrained. We now illustrate numerically  its ability to reproduce similar shapes of implied volatility surfaces and at-the-money skews for different values of the Hurst parameter, when calibrating both the parameter $H$ and the reversionary time-scale $\epsilon$.\\

\noindent For this, we first generate implied volatility surfaces of the hyper-rough and rough Heston model via the Fourier-Cosine series expansion technique from \cite{fangCosineExpansion}, where we used the fractional Adams scheme described in \cite{el2019characteristic} on the fractional Riccati equation \eqref{eq:VolterraRiccati} to compute the characteristic function of the (hyper-)rough Heston models. Three target smiles are generated with a (hyper-)rough Heston having parameters
\begin{equation} \label{eq:roughHestonParams}
    \rho=-0.7, \quad \theta=0.02, \quad \xi=0.3, \quad U_0=0.02,
\end{equation}
for $ H \in \left\{ 0.1, 0, -0.05 \right\}$, where we used the parametrization $\bar{G}_0(t) := \int_0^t g_0(s) ds$, with $g_0(t) = U_0 + \theta \int_0^t K_H(s)ds$.

\noindent For each of these smiles, we  calibrate the parameters $\left( \hat \epsilon, \hat H \right)$ of the reversionary Heston model \eqref{eq:reversionaryHestonInstantaneousSpot}-\eqref{eq:reversionaryHestonInstantaneousVol},  while fixing the other parameters equal to those of the hyper-rough Heston's, by minimizing a weighted loss
\begin{equation} \label{eq:calibrationLoss}
    \sum_{i,j} w_{i,j} \left( C_{\textit{hyper-rough Heston}}(T_i,k_j)-C^{\hat\epsilon, \hat H}_{\textit{reversionary Heston}}(T_i,k_j) \right)^2,
\end{equation}
where, for an initial underlying value $S_0=1$, at given maturity $T$ and strike $K$, we define the log-moneyness $k := \log \left( \frac{K}{S_0} \right)$, while $C_{\textit{hyper-rough Heston}}(T,k)$ and $C^{\hat\epsilon, \hat H}_{\textit{reversionary Heston}}(T,k)$ denote respectively the call prices computed in the hyper-rough Heston model for different values of the Hurst index and parameters \eqref{eq:roughHestonParams} and in the reversionary Heston model with parameters $\left( \hat \epsilon, \hat H \right)$ and \eqref{eq:roughHestonParams}. The reversionary Heston prices are also obtained by Fourier-Cosine expansion of the characteristic function. In contrast to the rough Heston models, the characteristic function is known explicitly, see Corollary \ref{C:marginal} below.  
After calibration, we obtain the following parameters
\begin{table}[H] \label{tab:calibrated_H_eps}
\centering
\begin{tabular}{|c|cc|}
\hline
\textbf{Target (hyper-)rough Heston} & \multicolumn{2}{c|}{\textbf{Calibrated reversionary Heston}} \\ \hline
$H$                          & \multicolumn{1}{c|}{$\hat \epsilon$}                 & $\hat H$        \\ \hline
0.1                          & \multicolumn{1}{c|}{0.10183756}         & -0.29183935        \\ \hline
0                            & \multicolumn{1}{c|}{0.06258637}         & -0.33057822        \\ \hline
-0.05                        & \multicolumn{1}{c|}{0.05932449}         & -0.38692275        \\ \hline
\end{tabular}
\caption{Calibrated values of $\left( \hat \epsilon, \hat H \right)$ of the reversionary Heston model to the hyper-rough Heston volatility surfaces. The other parameters are fixed as in \eqref{eq:roughHestonParams}.}
\end{table}

\noindent The resulting at-the-money skews between 1 week and 1 year are shown on Figure~\ref{fig:ATMskewsRoughHestonRevHeston}. The implied volatility surfaces for the case $H=0.1$ is illustrated on Figure~\ref{fig:smilesRoughHestonRevHeston}. The fit of the smiles for $H=0$ and $H=-0.05$ are deferred to Appendix \ref{ss:additional_plots}, see Figures \ref{fig:smilesRoughHestonH0RevHeston} and \ref{fig:smilesRoughHestonHminus0dot05RevHeston}. The graphs show that the reversionary Heston model seems to be able to generate similar shapes of the implied volatility surfaces of rough and hyper-rough models and very steep skews even in the hyper rough regimes $H\leq 0$.

\begin{figure}[H]
\begin{center}
\includegraphics[scale=0.3,angle=0]{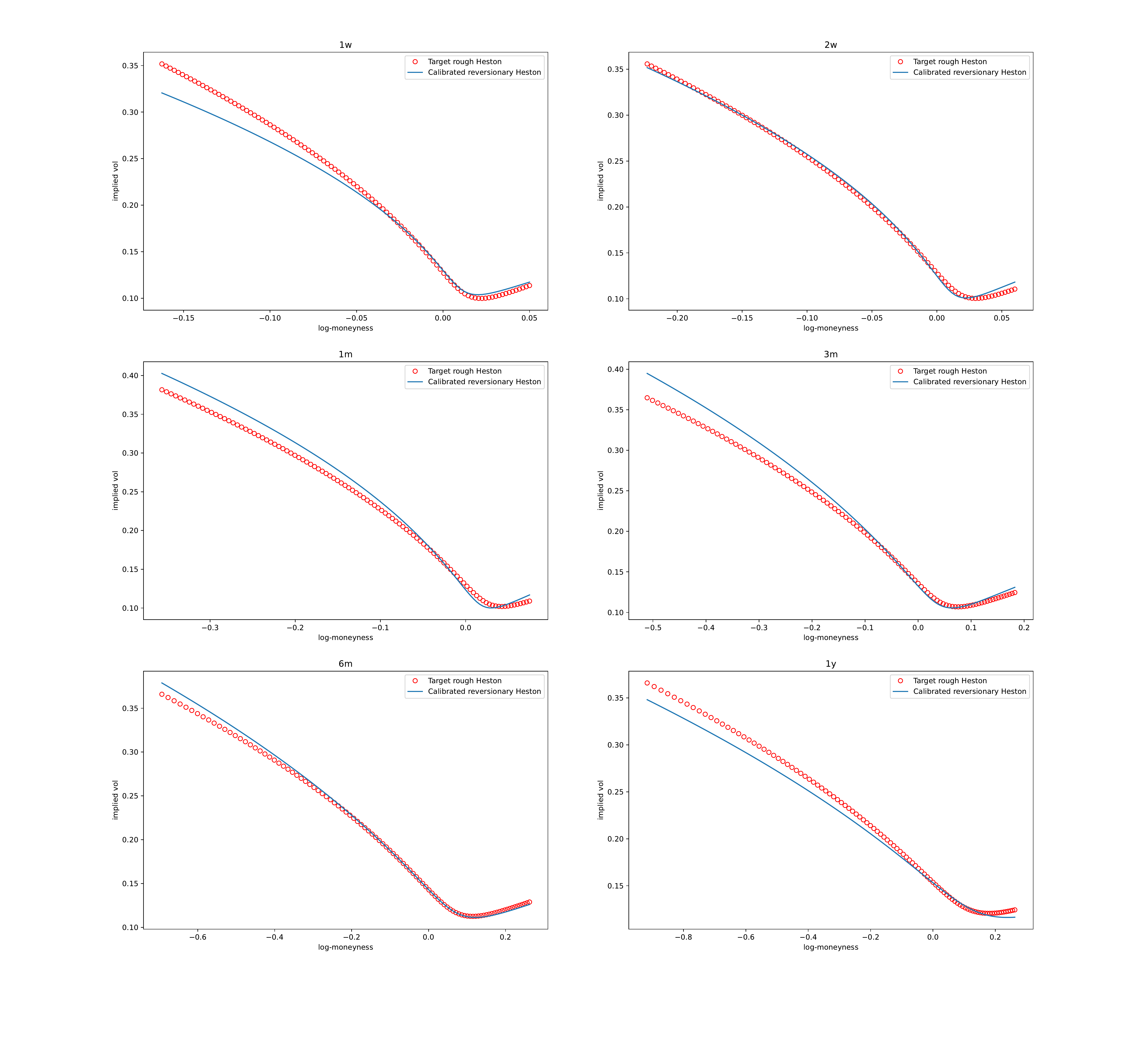}
\vspace*{-0.2in}
\caption{Smiles comparison between target rough Heston with parameters \eqref{eq:roughHestonParams}, with $H=0.1$, and reversionary Heston with calibrated parameters from the first row of Table \eqref{tab:calibrated_H_eps} for different maturities from one week to one year.}
\label{fig:smilesRoughHestonRevHeston}
\end{center}
\end{figure}

\begin{figure}[H]
\begin{center}
\includegraphics[width=6.2in,angle=0]{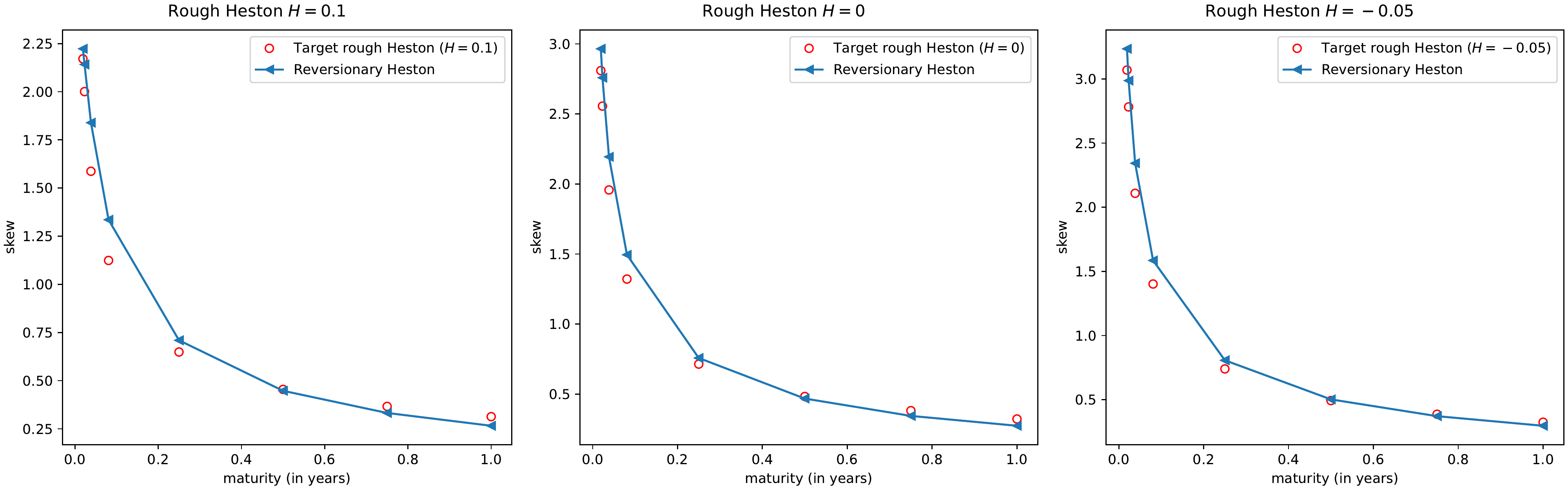}
\caption{Resulting at-the-money skew $\left\{\left| \partial_{k} \sigma_{\textit{implicit}} \left( k,T \right) \right|_{k=0}\right\}_{T}$ comparison between target rough Heston with parameters \eqref{eq:roughHestonParams} and reversionary Heston with calibrated parameters given in Table \eqref{tab:calibrated_H_eps} for different maturities from one week to one year.}
\label{fig:ATMskewsRoughHestonRevHeston}
\end{center}
\end{figure}

\noindent The whole section emphasizes how the reversionary Heston model \eqref{eq:reversionaryHestonInstantaneousSpot}--\eqref{eq:reversionaryHestonInstantaneousVol} can be seen as an \textit{engineering proxy} of the hyper-rough Heston model: it is obtained as a semi-martingale and Markovian proxy model such that for specific values of the couple $\left( \epsilon, H \right)$, it is able to mimic the skew behavior of hyper-rough Heston.

\begin{remark}[An engineering proxy of rough models]
    For a fixed $\epsilon > 0$, we can bound the error between call prices given by the reversionary Heston and the hyper-rough Heston in the same spirit of \cite[Proposition~4.3]{abi2019multifactor}: there exists a constant $c>0$ independent of $\epsilon$ such that
    \begin{equation}
        \left| C^{H}_{\textit{hyper-rough Heston}}(T,k) - C^{\hat H, \epsilon}_{\textit{reversionary Heston}}(T,k) \right| \leq c \int_0^T \left| K_H(s) - \hat  K_{\hat  H, \epsilon} (s) \right| ds,
    \end{equation}
    where  $K_H$ is the fractional kernel given in \eqref{eq:fractionalkernel} and  $\hat K_{\tilde H, \epsilon}(t) := \epsilon^{\hat H-1/2}e^{-\frac{t}{\epsilon}}$. Note that for fixed $H$, as $\epsilon \to 0$,  $\hat K_{H,\epsilon}$ does not converge to $K_H$ in $L^1$. In this sense, we cannot expect the reversionary model to converge towards the rough model as $\epsilon \to 0$, which can alternatively be seen as the impact of dropping the non-Markovian term in \eqref{eq:K'Z}. More precisely, we will explore in Section \ref{S:convjumps} what neglecting the non-Markovian term in \eqref{eq:K'Z} entails asymptotically for the reversionary Heston with $H \in (-1/2,1/2]$: it will converge to Black-Scholes, see Corollary~\ref{Cor:revHestonMarginals} and Figure~\ref{F:cf_convergence_all_regimes} (upper plots).  Having said that, one would like to calibrate $(\hat H, \epsilon)$ as done in Table~\ref{tab:calibrated_H_eps} above, or alternatively by mimizing the $L^1$-norm between the kernels $K_H$ and $\hat K_{\hat H, \epsilon}$, in order to make the reversionary model the closest possible to the rough model.
\end{remark}

\section{The joint characteristic functional of reversionary Heston}\label{S:jointchar}
{The following theorem provides the joint conditional characteristic functional of the log-price $\log S^{\epsilon}$} and the integrated variance $\bar{V}^\epsilon := \int_0^{\cdot} V^\epsilon_s ds$ in the model  \eqref{eq:reversionaryHestonInstantaneousSpot}--\eqref{eq:reversionaryHestonInstantaneousVol} in terms of a solution to a system of {time-dependent Riccati  ordinary differential equations.}

\begin{theorem}[Joint characteristic functional]\label{T:charfuneps}
     Let $f,g: [0,T]\to  \mathbb{C}$ be measurable and bounded functions  such that
\begin{equation}\label{eq:condfg}
    \Re g + \frac{1}{2} \left((\Re f)^2 -\Re f\right)  \le 0.
\end{equation}
Then, the joint conditional characteristic functional of $(\log S^{\epsilon}, \bar V^{\epsilon})$ is given by 
\begin{equation} \label{genericConditionalExpectation}
 \mathbb{E} \left[ \left. \exp\left(\int_t^T f({T-s}) d\log S^\epsilon_s + \int_t^T g({T-s}) d \bar{V}_s^\epsilon\right) \right| \mathcal{F}_t \right] = \exp\left(\phi_{\epsilon} \left( T-t \right) + \epsilon^{\frac{1}{2}-H} \psi_{\epsilon} \left( T-t \right) V^\epsilon_t \right), \quad t \leq T,
\end{equation}
where $(\phi_{\epsilon}, \psi_{\epsilon})$ is the solution to the following system of time-dependent Riccati equations 
\begin{align}
\phi_\epsilon'(t) & = \left( \theta + \epsilon^{-H-\frac{1}{2}} V_0 \right) \psi_\epsilon(t), \quad \phi_{\epsilon}(0) = 0, 
 \label{eq:Ric1}\\
   \psi_\epsilon'(t) & =  \epsilon^{H-\frac{1}{2}} \frac{\xi^2}{2}  \psi_\epsilon ^2(t) +  \left( \rho \xi \epsilon^{H - \frac{1}{2}} f(t) - \epsilon^{-1} \right) \psi_\epsilon(t) \nonumber  \\ 
   &\qquad \qquad + \epsilon^{H - \frac{1}{2}} \left( g({t}) + \frac{f^2({t})-f({t})}{2} \right), \qquad  \psi_\epsilon(0) = 0. \label{eq:Ric2}
\end{align}
\end{theorem}

\begin{proof}
    The proof is given in Section~\ref{S:proofchar} below.
\end{proof}
\noindent Before proving the result, we note that in the case $f$ and $g$ are constant, one recovers the usual formula for the characteristic function of the \citet{heston1993}  model, where the solution $(\phi_{\epsilon}, \psi_{\varepsilon})$ of \eqref{eq:Ric1}-\eqref{eq:Ric2} is explicit as stated in the following corollary.

\begin{corollary}[Explicit marginals]\label{C:marginal}
 Let $u,v \in \mathbb R$ and set 
$f(t) =i u$ and $g(t)=i v$, for all $t\in  [0,T]$. Then, the solution $(\phi_{\epsilon},\psi_{\epsilon})$  to the Riccati equations \eqref{eq:Ric1}-\eqref{eq:Ric2} is explicitly given by 
\begin{align*}
    \phi_{\epsilon}(t) &= \left( \epsilon^{-\frac{1}{2}-H} \theta + \epsilon^{-1-2H} V_0\right) \xi^{-2} \left( \left( 1 - i \rho \epsilon^{H+\frac{1}{2}} \xi u - d \right) t - 2 \epsilon \ln \left( \frac{1 - ge^{- \epsilon^{-1} td}}{1-g}\right) \right), \\
    \psi_{\epsilon}(t) &= \epsilon^{-H-\frac{1}{2}}\xi^{-2} \left( 1 - i \rho \epsilon^{H+\frac{1}{2}} \xi u - d \right) \frac{1 - e^{- \epsilon^{-1} td }}{1 - ge^{- \epsilon^{-1} td }} ,
\end{align*}
with
\begin{align}\label{eq:g_variable}
    g & := \frac{1 - i \rho \epsilon^{H+\frac{1}{2}} \xi u - d}{1 - i \rho \epsilon^{H+\frac{1}{2}} \xi u + d}, \quad 
    d  := \sqrt{\left( 1 - i \rho \epsilon^{H+\frac{1}{2}} \xi u \right)^2 - 2 \left( \epsilon^{H+\frac{1}{2}} \xi \right)^2 \left( i v - \frac{u^2 + iu}{2} \right)}, \quad \Re d > 0.
\end{align}
 Consequently, the conditional joint characteristic function  of $(\log S_T^{\epsilon}, \int_t^T V^{\epsilon}_s ds )$ is given by 
\begin{equation} \label{eq:CHreversionaryHeston}
 \mathbb{E} \left[ \left. \exp\left( iu \log {S^\epsilon_T} + iv\int_t^T   {V}_s^\epsilon ds\right) \right| \mathcal{F}_t \right] = \exp\left(iu\log S^{\epsilon}_t +\phi_{\epsilon} \left( T-t \right) + \epsilon^{\frac{1}{2}-H} \psi_{\epsilon} \left( T-t \right) V^\epsilon_t \right), \quad t \leq T.
\end{equation}
\end{corollary}

\begin{proof}
    For the explicit derivation of the formulas, see for example  \cite[Chapter 2]{volSurfaceGatheral}. An application of Theorem~\ref{T:charfuneps} yields the result.
\end{proof}

\begin{remark}
Such formulas for $\phi_{\epsilon}$ avoid branching issues as described in \cite{LittleHestonTrap}.
\end{remark}

\noindent The rest of the section if dedicated to the proof of Theorem~\ref{T:charfuneps}. We first study the existence of a solution to time-dependent Riccati ODEs for which equation \eqref{eq:Ric2} is a particular case, and provide some of their properties in Section~\ref{S:riccatiexistence}. We complete the proof of Theorem~\ref{T:charfuneps} in Section~\ref{S:proofchar}. 

\subsection{Time-dependent Riccati ODEs: existence and uniqueness}\label{S:riccatiexistence}

 In this section, we consider a generic class of time-dependent Riccati equations that encompass 
 equation \eqref{eq:Ric2},  in the form 
\begin{equation} \label{eq:IVP}
    {\psi'}(t)  = 
a(t) \psi^2(t) + b(t) \psi(t) + c(t), \quad \psi(0) = u_0, \quad t\le T,
\end{equation}
 with $u_0 \in \mathbb C$ and $a,b,c:[0,T]\to \mathbb C$ three measurable and bounded functions. We say that $\psi:[0,t^*] \to \mathbb C$ for some $t^* \in \left( 0,T \right]$ is a local extended solution to \eqref{eq:IVP} with some initial condition $\psi(0)=u_0 \in \mathbb C$ if, almost everywhere on $[0,t^*]$, it is continuously differentiable and satisfies the relations in \eqref{eq:IVP}. The extended solution is global if $t^*=T$.\\  

\noindent The presence of the squared non-linearity in \eqref{eq:IVP} precludes the application of the celebrated Cauchy-Lipschitz theorem and  can lead to explosive solutions in finite time. Compared to the related literature on similar Riccati equations  \cite[Lemma 2.3 and Section B]{Fili2009AffineDiffusion}, we provide a concise and simplified proof for the existence and uniqueness of a global extended solution to the Riccati equation \eqref{eq:IVP} using a variation of constant formula under the following assumption on the coefficients $(a,b,c)$ and the initial condition $u_0$:
\begin{equation}\label{eq:assumptionscoeff}
    \Im(a(t)) = 0,  \quad a(t) >0, \quad  \Re\left( c(t) \right) + \frac{\Im \left( b(t) \right)^2}{4a(t)} \le 0, \quad   \Re(u_0)\leq 0, \quad t\leq T. 
\end{equation}

\noindent The following theorem gives the existence and uniqueness of a solution to the Riccati equation \eqref{eq:IVP}.

\begin{theorem}[Existence and uniqueness for the Riccati] \label{existenceTheorem}
Let $u_0\in \mathbb C$ with $\Re(u_0)\leq 0$ and $a,b,c:[0,T]\to \mathbb C$ be measurable and bounded functions satisfying \eqref{eq:assumptionscoeff}. Then,  there exists a unique extended solution $\psi:[0,T]\to \mathbb C$ to the Riccati equation \eqref{eq:IVP} such that 
\begin{align}\label{eq:psireal}
   \Re (\psi(t))\leq 0, \quad t\leq T,
\end{align}
and
\begin{align}\label{eq:psibound}
    \sup_{t\leq T} |\psi(t)|<\infty.
\end{align}
\end{theorem}

\begin{proof} {For the existence part, we proceed in two steps. First, we  start by arguing the existence of a local solution using 
Carathéodory's theorem.  For this we rely on \cite[Chapter 12, Section 2]{GripenbergVolterraIntegral},  using the notations therein (see equation (1.7) for example), we consider the integral equation
\begin{equation} \label{eq:GripenbergEquation}
    \psi(t) = \psi(0) + \int_0^t g \left(t,s,\psi(s)\right)ds, \quad t \geq 0,
\end{equation}
where the operator $g$ is defined by
\begin{equation*}
    g(t,s,\psi(s)):= a(s) \psi(s)^2 + b(s) \psi(s) + c(s).
\end{equation*}
 Let $D$ be an open, connected subset of $\mathbb R^+ \times \mathbb C$ that contains $\left( 0, \psi(0) \right)$, and $C\left( D, \mathbb C \right)$ be the set of continuous applications valued from $D$ to $\mathbb C$. Define
\begin{equation*}
    T_{\infty} := \sup \left\{ t \in \mathbb R^+ | C_{\psi(0),D} \left( [0,T), \mathbb C \right) \neq \emptyset \right\},
\end{equation*}
where
\begin{equation*}
    C_{\psi(0),D} \left( [0,T), \mathbb C \right) := \left\{ \phi \in C\left( D, \mathbb C \right) | \phi(0) = \psi(0) \text{ and } \left( t, \phi(t)\right) \in D \text{ for } t \in [0,T)\right\}.
\end{equation*}
\noindent An application of \cite[Theorem 2.6]{GripenbergVolterraIntegral} yields the existence of a unique non-continuable solution to \eqref{eq:GripenbergEquation} which means that $\left( t, \psi(t) \right) \in D$ on the interval $[0,T_{\infty})$ and that either $T_{\infty} = T$ or $\underset{t \rightarrow T_{\infty}}{\text{lim}} |\psi(t)| = \infty$.  Indeed, the assumptions (i) to (v) of \cite[Theorem 2.6]{GripenbergVolterraIntegral} are readily satisfied by boundedness and integrability of $a$, $b$ and $c$ and the fact that $g$ does not depend on $t$ and satisfies the Carathéodory conditions.}
\noindent {Second,} we argue that 

\begin{equation}\label{eq:suppsit+}
      \sup_{t\leq T_{\infty}} |\psi(t)|<\infty,
\end{equation} 
which would then yield $T_{\infty}=T$ and the existence of a global solution $\psi$.  Let $t< T_{\infty}$. We start by showing that $\Re(\psi(t)) \leq 0$.  Indeed, taking real parts in \eqref{eq:IVP}, $\psi_{\text{\bold r}}:=\Re \left( \psi \right)$ satisfies the following equation on $[0,T_{\infty})$:
    \begin{equation} \notag
       \psi_{\text{\bold r}}' (s)= \left\{ a(s) \Re \left( \psi(s) \right) + \Re \left( b(s) \right) \right\} \psi_{\text{\bold r}}(s)+ d(s),
    \end{equation}
where $d(s) = - a(s) \left( \Im\left( \psi(s) \right) + \frac{\Im \left( b(s) \right)}{2a(s)} \right)^2 + \Re \left( c(s) \right) + \frac{\Im^2\left( b(s) \right)}{4a(s)} \le 0$ thanks to condition \eqref{eq:assumptionscoeff}, after a completion of squares. The variation of constant for $\psi_{\text{\bold r}}$ then yields
    \begin{equation} \notag
       \psi_{\text{\bold r}}(t) = e^{\int_0^t \left( a(u) \Re \left( \psi(u) \right) + \Re\left( b(u) \right) \right) du} \Re(u_0)  + \int_0^t   d(s)  e^{\int_s^t \left( a(u) \Re \left( \psi(u) \right) + \Re \left( b(u) \right) \right) du}ds \le 0,
    \end{equation}
 since the exponential is positive and $d(s)\leq 0$, and $\Re(u_0)\leq 0$ by assumption.  This shows that $\Re(\psi)\leq 0,  $ on $[0,T_{\infty})$. Finally, an application of a similar variation of constants formula on  equation \eqref{eq:IVP} leads to 
  \begin{equation} \notag
       \psi(t) = e^{\int_0^t \left( a(u) \left( \psi(u) \right) + \Re\left( b(u) \right) \right) du} u_0  + \int_0^t  e^{\int_s^t \left( a(u) \left( \psi(u) \right) +  \left( b(u) \right) \right) du} c(s) ds,
    \end{equation}
so that taking the module together with the triangle inequality and the fact that $\Re(\psi)\leq 0 $ on $[0,T_{\infty})$, yields 
\begin{align*}
|\psi(t)| &\leq  \left |e^{\int_0^t \left( a(u)  \psi(u)  + \left( b(u) \right) \right) du} \right | |u_0|  + \int_0^t  \left| e^{\int_s^t \left( a(u)  \psi(u)+  \left( b(u) \right) \right) du}\right|  |c(s)| ds\\
&=    e^{\int_0^t \left( a(u)  \Re \left( \psi(u) \right) + \Re\left( b(u) \right) \right) du}  |u_0|  + \int_0^t  e^{\int_s^t \left( a(u)  \Re\left( \psi(u) \right) +  \Re\left( b(u) \right) \right) du} |c(s)|ds\\
&\leq  e^{\int_0^t  \Re\left( b(u) \right)  du}  |u_0|  + \int_0^T  e^{\int_s^t  \Re\left( b(u) \right)  du}  |c(s)|ds \\
&\leq C\left( |u_0| + \int_0^T |c(s)|ds\right),
\end{align*}
where $C=\sup_{s,s'\in [0,T]^2} e^{\int_s^{s'} \Re(b(u)) du}$  does not depend on $t$ 
 and is finite by boundedness of $b$. 
This shows \eqref{eq:suppsit+} as needed. Combining the above we obtain the existence of a solution $\psi$ on $[0,T]$ satisfying \eqref{eq:psireal} and \eqref{eq:psibound}.

{
\noindent To argue uniqueness, assume there are two such extended solutions $\psi_1$ and $\psi_2$ that satisfy \eqref{eq:IVP}. Then,
\begin{equation*}
    \left(\psi_2-\psi_1\right)'(t) = 
\left( a(t) \left(\psi_2+\psi_1\right)(t) + b(t) \right) \left(\psi_2-\psi_1\right)(t), \quad \left(\psi_2-\psi_1\right)(0) = 0, \quad t\le T,
\end{equation*}
which yields
\begin{equation*}
    \left|\psi_2-\psi_1\right|(t) \le 
\int_0^t \left| a(s) \left(\psi_2+\psi_1\right)(s) + b(s) \right| \left|\psi_2-\psi_1\right|(s) ds \leq c \int_0^t \left|\psi_2-\psi_1\right|(s) ds ,  \quad t\le T,
\end{equation*}
for some $c>0$ by boundedeness of $(\psi_1,\psi_2, a, b)$ using \eqref{eq:suppsit+}, so that 
the uniqueness is obtained  from Gronwall's lemma.
}
\end{proof}

\subsection{Proof of Theorem~\ref{T:charfuneps}}\label{S:proofchar}
We first argue the existence of a solution to the system of Riccati equations \eqref{eq:Ric1}-\eqref{eq:Ric2}. Let us rewrite the Riccati ODE from \eqref{eq:Ric2} as
\begin{equation} \label{eq:reformulation_ODE_psi_eps}
    \psi_\epsilon'(t) = a_\epsilon \psi_\epsilon ^2(t) + b_\epsilon(t) \psi_\epsilon(t) +c_\epsilon(t), \quad \psi_\epsilon(0) = 0, \quad t \le T,
\end{equation}
where we defined
\begin{equation}\label{eq:abc}
\begin{cases} 
    a_\epsilon & := \epsilon^{H-\frac{1}{2}} \frac{\xi^2}{2}\\
    b_\epsilon(t) & := \epsilon^{H-\frac{1}{2}} \rho \xi f(t) -  \epsilon^{-1} \\
    c_\epsilon(t) & := \epsilon^{H-\frac{1}{2}} \left[ g(t) + \frac{f^2(t) - f(t)}{2} \right].
\end{cases}
\end{equation}
Since condition \eqref{eq:condfg} ensures
$$
\Re\left( c_{\epsilon} \right) + \frac{\Im \left( b_{\epsilon} \right)^2}{4a_{\epsilon}} = \epsilon^{H-1/2} \left( \Re g + \frac{1}{2} \left( \left(\Re f\right)^2 - \Re f \right) + \left( \rho^2 - 1 \right) \left( \Im f \right)^2 \right) \leq 0,
$$
then conditions \eqref{eq:assumptionscoeff} are readily satisfied and consequently Theorem \ref{existenceTheorem} yields the existence and uniqueness of a solution $\psi_\epsilon:[0,T]\to \mathbb C$
to the Riccati ODE \eqref{eq:Ric2} such that
$$
\Re(\psi_\epsilon(t)) \leq 0,\quad t \leq T.
$$
 The function  $\phi_{\epsilon}$ defined  in integral form as
\begin{equation*}
    \phi_{\epsilon}(t) = \left( \theta + \epsilon^{-H-\frac{1}{2}} V_0 \right) \int_0^t \psi_\epsilon(s)ds, \quad t \le T,
\end{equation*}
solves  \eqref{eq:Ric1}. 

\noindent We now prove the expression for the charateristic functional \eqref{genericConditionalExpectation}. Define the following process $M$:
\begin{align*}
M_t &= \exp(U_t),  \\
U_t &=  \phi_{\epsilon} \left( T-t \right) + \epsilon^{\frac{1}{2}-H} \psi_{\epsilon} \left( T-t \right) V^\epsilon_t +    \int_0^t f(T-s) d\log S^\epsilon_s + \int_0^t g(T-s) d \bar{V}_s^\epsilon.
\end{align*}
In order to obtain \eqref{genericConditionalExpectation}, it suffices to show that $M$ is a martingale. Indeed, if this is the case, and after observing that the terminal value of $M$  is  given by
$$ M_T = \exp\left(  \int_0^T f(T-s) d\log S^\epsilon_s + \int_0^T g(T-s) d \bar{V}_s^\epsilon\right),$$
recall that $\phi_{\epsilon}( 0 )=\psi_{\epsilon}( 0 )=0$,
we obtain 
\begin{align*}
 \mathbb E\left[ \exp\left(  \int_0^T f(T-s) d\log S^\epsilon_s + \int_0^T g(T-s) d \bar{V}_s^\epsilon\right) \bigg|\mathcal F_t\right] =  \mathbb E\left[  M_T  \big| \mathcal F_t\right] = M_t = \exp\left(U_t\right),
\end{align*}
which yields \eqref{genericConditionalExpectation}. We now argue that $M$ is a martingale. We first show that $M$ is a local martingale using Itô formula. The dynamics of $M$ read
\begin{align*}
    dM_t = M_t \left(dU_t + \frac 1 2 d\langle U\rangle_t \right),
\end{align*}
with 
\begin{align*}
     dU_t = & \left\{ {\bl - } \phi_{\epsilon}'(T-t) {\bl + } \left( \theta + \epsilon^{-\frac{1}{2}-H} V_0 \right) \psi_\epsilon(T-t)  + \left( - \epsilon^{\frac{1}{2}-H} \psi_\epsilon '(T-t) - \epsilon^{-\frac{1}{2}-H} \psi_\epsilon(T-t) + g(T-t) - \frac{f(T-t)}{2}\right) V^\epsilon_t \right\} dt\\
    & + \left( \xi \psi_\epsilon(T-t) + \rho f(T-t) \right) \sqrt{V^\epsilon_t}dW_t+\sqrt{1-\rho^2}f(T-t)\sqrt{V^\epsilon_t}dW^\perp_t.
\end{align*}
This yields that the drift in $dM_t/M_t$ is given by 
\begin{align*}
    & {\bl - } \phi_\epsilon'(T-t) {\bl + } \left( \theta + \epsilon^{-\frac{1}{2}-H} V_0 \right) \psi_\epsilon(T-t) \\
    & + \left( - \epsilon^{\frac{1}{2}-H} \psi_\epsilon'(T-t) + \frac{\xi^2}{2} \left( \psi_\epsilon (T-t) \right)^2 + \left( \rho \xi f(T-t) - \epsilon^{-\frac{1}{2}-H} \right) \psi_\epsilon(T-t) + g(T-t) + \frac{f^2(T-t) - f(T-t)}{2}\right) V^\epsilon_t
\end{align*}

\noindent which is equal to $0$ from the Riccati equations \eqref{eq:Ric1} and \eqref{eq:Ric2}. This shows that $M$ is a local martingale. 
{To argue that $M$ is a true martingale, we note that $\Re(\psi_{\epsilon})\leq 0$ which implies $\Re(\phi_{\epsilon})\leq 0$, so that 
\begin{align*}
    \Re(U_t) &\leq \int_0^t \Re(f(T-s)) d\log S^\epsilon_s + \int_0^t \Re(g(T-s)) d \bar{V}_s^\epsilon \\
    &= \int_0^t \left( \Re(g(T-s)) 
 -\frac 12\Re(f(T-s))\right)  {V}_s^\epsilon ds + \int_0^t   \Re(f(T-s))  \sqrt{{V}_s^\epsilon}dB_s \\
 &\leq  -\frac 12 \int_0^t  
\Re(f(T-s))^2 {V}_s^\epsilon ds + \int_0^t   \Re(f(T-s))  \sqrt{{V}_s^\epsilon}dB_s=:\tilde U_t,
\end{align*}
where the last inequality follows from \eqref{eq:condfg}.  It follows that 
\begin{align*}
    |M_t|=\exp(\Re(U_t)) \leq \exp(\tilde U_t),
\end{align*}
where the process $\exp(\tilde U)$ is a true martingale, see \cite[Lemma 7.3]{abi2019affine}. This shows that  $M$ is a true martingale, being a local martingale bounded by a true martingale, see \cite[Lemma 1.4]{jarrow2018continuous}, which concludes the proof. 
} 


\section{From reversionary Heston to jump processes}\label{S:convjumps}

In this section, we establish the convergence of the log-price and the integrated variance $(\log S^{\epsilon},\bar V^{\epsilon})$ in  the reversionary Heston model \eqref{eq:reversionaryHestonInstantaneousSpot}-\eqref{eq:reversionaryHestonInstantaneousVol} towards a Lévy jump process $(X,Y)$, as $\epsilon$ goes to $0$. More precisely, the limit $(X, Y)$ belongs to the class of Normal Inverse Gaussian - Inverse Gaussian (NIG-IG) processes defined as follows.

\begin{definition}[NIG-IG process]\label{D:NIGIG} Fix  $\alpha \geq |\beta| \geq 0$, $\delta,\lambda > 0$ and $\mu \in \R$. We say that $\left( X_t, Y_t \right)_{t \geq 0}$ is a Normal Inverse Gaussian - Inverse Gaussian (NIG-IG) process with parameters $\left( \alpha,\beta,\delta,\mu,  \lambda\right)$ if it is a two-dimensional homogeneous Lévy process with càdlàg sample paths, starting  from $(X_0,Y_0)=(0,0)$ almost surely, with Lévy exponent $\eta$ defined by
\begin{equation} \label{eq:eta_definition}
    \eta(u,v) := \left[ i \mu u + \delta \left( \sqrt{\alpha^2 - \beta^2} -\sqrt{\alpha^2 - 2i\lambda v -\left( \beta + iu \right)^2} \right) \right], \quad u, v \in \mathbb{R},
\end{equation}
i.e.~the joint characteristic function is given by
\begin{equation} \label{eq:charnigigprocess}
    \mathbb{E}\left[ \exp \left( iu X_t + iv Y_t \right) \right] = \exp \left( \eta(u,v) t \right), \quad u, v \in \mathbb{R}, \quad t \leq T.
\end{equation}
\end{definition}

\noindent In order to justify the existence of such a class of Lévy processes, one needs to justify that $\eta$ given in \eqref{eq:eta_definition} is indeed the logarithm of a characteristic function associated to an infinitely divisible distribution, see  \cite[Corollary 11.6]{sato_levy_processes}.  
 This is the object of the following lemma, which also provides the link with first-hitting times and subordinated processes. 

\begin{lemma}[Representation using subordination]\label{L:representation_subordination}
    Let $\alpha \geq |\beta| \geq 0$, $\delta,\lambda > 0$, $\mu \in \R$  and   $(\widetilde W, \widetilde W^{\perp})$ be a two dimensional Brownian motion. Let $(\Lambda_t)_{t\in [0,T]}$ be the first hitting-time process defined as
    \begin{equation}\label{eq:firstpassagetime}
    \Lambda_t :=  \inf \left\{ s \geq 0 : \sqrt{\alpha^2-\beta^2}s + \widetilde W_{s} \geq \delta t \right\}, \quad t \in \left[ 0,T 
    \right],
\end{equation}
and define $Z$ as the following shifted subordinated process 
\begin{align}\label{eq:subord}
    Z_t =\mu t + \beta \Lambda_t +  \widetilde W^{\perp}_{\Lambda_t}, \quad t \in \left[ 0,T 
    \right].
\end{align}
Then,
\begin{align}\label{eq:charfunsub}
    \mathbb E\left[ \exp\left( iuZ_t+ iv \lambda \Lambda_t \right) \right]  = \exp(\eta(u,v)t), \quad u,v \in \mathbb R, \quad t \in \left[ 0,T 
    \right].
\end{align}
 In particular, $\eta$ given by \eqref{eq:eta_definition} is the logarithm of the characteristic function  of the joint random variable $(Z_1,\lambda \Lambda_1)$ which is infinitely divisible.
\end{lemma}
\begin{proof}
Fix $t \in \left[ 0,T \right]$. By construction, it is well-known that  $\Lambda_t$ has an Inverse Gaussian distribution if $\alpha > \left| \beta \right|$ with parameters $\textit{IG}\left( \frac{\delta t}{\sqrt{\alpha^2 - \beta^2}}, \delta^2 t \right)$, and in the drift-free case $\alpha = \left| \beta \right|$, $\Lambda_t$ follows a Lévy distribution with parameters $\textit{Lévy}\left( 0, \delta^2 t \right)$ (see \cite{barndorff_nig_1997} and Definition \ref{D:definitions_ig_levy_nig_ig} in the Appendix). Now conditional on $\Lambda_t$, $Z_t$ is Gaussian with parameters $\mathcal N \left( \mu t + \beta \Lambda_t, \Lambda_t \right) $ and using the tower property of conditional expectation, we get for the first case that
    \begin{align*}
        \mathbb E \left[\exp\left( iu Z_t + iv \lambda \Lambda_t \right) \right] & = \mathbb E \left[ \mathbb E \left[ \left. \exp\left( iu Z_t \right) \right| \Lambda_t \right] \exp \left( iv \lambda \Lambda_t \right) \right] \\
        & = \mathbb E \left[ \exp \left( iu \left( 
        \mu t + \beta \Lambda_t \right) - \frac{\Lambda_t u^2}{2} + iv \lambda \Lambda_t \right) \right] \\
        & = \exp \left( iu \mu t \right) \mathbb E \left[ \exp \left( \left( iu \beta - \frac{u^2}{2} + iv \lambda \right) \Lambda_t \right) \right] \\
        & = \exp \left( iu \mu t + \delta t \sqrt{\alpha^2 - \beta^2} \left( 1 - \sqrt{1-\frac{2}{\alpha^2 - \beta^2} \left( iu \beta - \frac{u^2}{2} + iv \lambda \right)} \right) \right)\\
        \text{i.e. }\mathbb{E} \left[ \exp \left( iuZ_t+iv\lambda \Lambda_t \right) \right] & = \exp \left( i u \mu t + \delta t \left( \sqrt{\alpha^2 - \beta^2} -\sqrt{\alpha^2 - 2i\lambda v -\left( \beta + iu \right)^2} \right)\right) = \exp(\eta(u,v)t), \quad u,v \in \mathbb R,
    \end{align*}
where we used Definition \ref{D:definitions_ig_levy_nig_ig} to get the fourth equality, noting that $\Re \left( i \left(u \beta + v \lambda \right) - \frac{u^2}{2} \right) \leq 0$. Similar computations yield the result for the case $\alpha = \left| \beta \right|$. Furthermore, we will say that that the random variable $(Z_1,\lambda \Lambda_1)$ follows a NIG-IG distribution with parameters $(\alpha,\beta,\mu, \delta,\lambda)$ (see Definition \ref{D:definitions_ig_levy_nig_ig} in the Appendix). Such distribution is  infinitely divisible because if $\left( X_1, Y_1 \right), \cdots, \left( X_m, Y_m \right)$ are independent NIG-IG random variables with common parameters $\left( \alpha, \beta, \lambda \right)$ and individual $\left( \mu_i, \delta_i \right)$, for $i =1, \cdots, m$, then $\left( X, Y \right) :=  \left( \sum_{i=1}^mX_i, \sum_{i=1}^mY_i \right)$ is again NIG-IG-distributed with parameters $\left( \alpha, \beta, \sum_{i=1}^m\mu_i, \sum_{i=1}^m\delta_i, \lambda \right)$.
\end{proof}

\noindent The appellation NIG-IG for the couple $(X,Y)$ in Definition~\ref{D:NIGIG} is justified as follows:
  \begin{itemize}
\item  
    $Y$ is an Inverse Gaussian process first derived by \citet{schrodinger_1915} which can be checked either by recovering the Inverse Gaussian distribution with parameters $\textit{IG}\left( \frac{\lambda \delta}{\sqrt{\alpha^2 - \beta^2}}, \lambda \delta^2 \right)$ after setting $u=0$ in \eqref{eq:charnigigprocess}; or by using  the representation as a first passage-time in \eqref{eq:firstpassagetime}. It is worth pointing that, 
 for $\alpha = |\beta|$, one recovers the well-known Lévy distribution for the  first-passage of a Brownian motion  with parameters $\textit{Lévy}\left( 0, \lambda \delta^2 \right)$. The \textit{Lévy} distribution can be seen as a special case of the \textit{Inverse Gaussian} distribution. 
    \item 
    $X$ is the celebrated Normal Inverse Gaussian process of \citet{barndorff_nig_1997}, with parameters $\textit{NIG}\left( \alpha, \beta, \mu, \delta \right)$, which can be checked by setting $v=0$ in \eqref{eq:charnigigprocess} or by using the representation as subordinated Brownian motion with an Inverse Gaussian subordinator as in \eqref{eq:subord}.
\end{itemize}

\noindent In addition, we allow in Definition~\ref{D:NIGIG} the parameter $\alpha$ to be equal to $\infty$ in the following sense:

\begin{remark}[Normal process]\label{R:Norma}
Considering the set of parameters
    \begin{equation*}
        (\alpha,\beta,\delta,\mu, \lambda) =  \left( \alpha, 0, \sigma^2 \alpha, \mu, 1 \right),
    \end{equation*}
    a second order Taylor expansion, as $\alpha \to \infty$, of the square root yields
    \begin{equation*}
        \mathbb{E}\left[ \exp \left( iu X_t + iv Y_t \right) \right] = \exp \left( \left[ i \mu u - \sigma^2 \left(\frac{u^2}{2} + iv\right)\right] t \right), \quad u, v \in \mathbb{R}, \quad t \leq T,
    \end{equation*}
    which is equivalent to the normal-deterministic process defined by
    \begin{equation*}
         \left( X_t, Y_t \right)_{t \in [0,T]} =  \left( \mu t + \sigma \widetilde W_t, \sigma^2 t \right)_{t \in [0,T]} \overset{d}{=} \left( \mu t + \widetilde W_{\sigma^2 t}, \sigma^2 t \right)_{t \in [0,T]},
    \end{equation*}
{where $\widetilde W$ is an $\mathcal{F}$-adapted Brownian motion.} We will consider that such (degenerate) process is a particular case of Definition~\ref{D:NIGIG} with parameters denoted by
    $$\left. \left( \alpha, 0, \sigma^2 \alpha, \mu, 1 \right)\right|_{\alpha \to \infty}.$$
\end{remark}

\noindent We are now in place to state our main convergence theorem. Theorem~\ref{T:charfunlimit} provides the convergence of the finite-dimensional distributions of the joint process $(\log S^{\epsilon}, \bar V^{\epsilon})$  through the  study of the limiting behavior of the  characteristic functional given in Theorem~\ref{T:charfuneps}. Interestingly, the limiting behavior disentangles three different asymptotic regimes based on the values of $H$ that can be seen intuitively on the level of the Riccati equation \eqref{eq:Ric2} as follows. Applying the variation of constants on $\psi_\epsilon$ from equations \eqref{eq:reformulation_ODE_psi_eps} and \eqref{eq:abc}, we get:
\begin{align}
    \psi_\epsilon(t) & = \epsilon^{H+1/2} \int_0^t K_\epsilon(t-s) F(s,\psi_{\epsilon}(s)) ds, \quad t\leq T,\label{eq:varconstantpsi} \\
    F(s,u) & := \frac{\xi^2}{2} u^2 + \rho \xi f(s) u + g(s) + \frac{f^2(s) - f(s)}{2}, \label{eq:varconstantF}
\end{align}
with $K_{\epsilon}$ the kernel defined by \begin{equation}\label{eq:kernelexp}
    K_\epsilon(t) =  {\epsilon}^{-1} e^{-{\epsilon}^{-1}t}, \quad  t\geq 0.
\end{equation}
Assuming that $\psi_{\epsilon}$ converges to some $\psi_0$ and observing that $K_{\epsilon}$ plays the role of the Dirac delta as $\epsilon \to 0$, one expects $\int_0^t K_\epsilon(t-s) F(s,\psi_{\epsilon}(s)) ds \to F(t,\psi_0(t))$ in \eqref{eq:varconstantpsi}, the  pre-factor $\epsilon^{H+1/2}$ suggests then  three different limiting regimes with respect to $H$ that can be characterized through the functions $F$  and  $\psi_0$: 
\begin{equation}
  \begin{cases}
    \psi_0(t) = 0,  &\mbox{if } H>-1/2,   \label{eq:psi0}\\
     \psi_0(t)= F(t,\psi_0(t)),  \quad  \Re(\psi_0(t))\leq 0,  &\mbox{if } H=-1/2,   \\
     0 = F(t,\psi_0(t)),  \quad 
 \quad 
 \, \, \; \,
 \Re(\psi_0(t))\leq 0,   &\mbox{if } H<-1/2.
 \end{cases}
\end{equation}
The function  $\psi_0$ in \eqref{eq:psi0} is even explicitly given by 
\begin{equation}
 \psi_0(t)=  \begin{cases}
    0, \quad &\mbox{if } H>-1/2,   \label{eq:psi00}\\
     \xi^{-2} \left( 1 - \rho \xi f(t) - \sqrt{\left( 1 - \rho \xi f(t) \right)^2 - 2 \xi^2 \left( g(t) + \frac{f^2(t)-f(t)}{2} \right)} \right),  &\mbox{if } H=-1/2,   \\
      - \xi^{-1} \left( \rho f(t) + \sqrt{f(t) \left( 1 - \left( 1- \rho^2 \right) f(t) \right) - 2 g(t)} \right),  &\mbox{if } H<-1/2,
 \end{cases}
\end{equation}
see Lemma~\ref{L:lemmaExistenceRootWithNegRealPart} below. {Furthermore, the convergence of the integrated variance process is strengthened to a functional weak convergence on the Skorokhod space  $D$ of real-valued càdlàg paths on $[0,T]$ endowed with the strong $M_1$ topology, see Section~\ref{ss:SM1_topology_remainder} below. Such topology is weaker and less restrictive than the commonly used uniform or $J_1$ topologies which share the property that a jump in a limiting
process can only be approximated by jumps of comparable size at the same time or, respectively, at nearby times. On the contrary, the $M_1$ topology of \citet{skorokhod_1956_m1} captures approximations of unmatched jumps, which in our case, will allow us to prove the convergence of the stochastic process $\bar V^{\epsilon}$  with continuous sample trajectories towards a Lévy process with càdlàg sample trajectories. The statement is now made rigorous in the following theorem.}

\begin{theorem}[Convergence towards NIG-IG processes]\label{T:charfunlimit}
    Let $f,g:[0,T]\to \mathbb C$ be bounded and measurable such that $\Re f = \Re g = 0$ and such that $\psi_0$ defined in \eqref{eq:psi00} has bounded variations. Then, based on the value of $H$, we obtain different explicit asymptotic formulas for the characteristic functional given in Theorem~\ref{T:charfuneps}:
    \begin{equation}\label{eq:convchar}
        \lim_{\epsilon \to 0} \mathbb{E} \left[ \exp\left(\int_0^T f(T-s) d\log {S^\epsilon_s} + \int_0^T g(T-s) d \bar{V}_s^\epsilon\right) \right]=\exp\left( \phi_0(T)\right),
    \end{equation}
    with 
     \begin{equation}
\phi_0(T) := \begin{cases}
    V_0 \int_0^T h(s) ds, \quad &\mbox{if } H>-1/2,   \label{eq:phi1} \\
    \left( \theta + V_0 \right) \xi^{-2} \left( T - \int_0^T \left( \rho \xi f(s) + \sqrt{\left( 1 - \rho \xi f(s)\right)^2 - 2 \xi^2 h(s)} \right) ds \right),  &\mbox{if } H=-1/2, \\
    -\theta \xi^{-1} \int_0^T \left( \rho f(s) + \sqrt{\rho^2 f^2(s) - 2 h(s)} \right) ds, \quad &\mbox{if } H<-1/2,
    \end{cases}
    \end{equation}
where $h(s): =g(s)+ \frac{f^2(s)-f(s)}{2}$.
In particular for $\rho \in (-1,1)$, as $\epsilon \to 0$,   the finite-dimensional distributions of the joint process $(\log \frac{S^{\epsilon}}{S_0},\bar V^{\epsilon})$ converge to the finite-dimensional distributions of a NIG-IG process $(X,Y)$ in the sense of Definition~\ref{D:NIGIG} with the following parameters depending on the value of $H$:
 \begin{equation}\label{eq:paramlimit}
(\alpha,\beta,\delta,\mu, \lambda) := \begin{cases}
    \left.\left( \alpha, 0, V_0 \alpha, - \frac{V_0}{2}, 1 \right)\right|_{\alpha \to \infty},  \quad &\mbox{if } H>-1/2, \\
 \left( \frac{1}{2} \frac{\sqrt{(\xi-2\rho)^2 + 4(1-\rho^2)}}{\xi (1-\rho^2)}, -\frac{1}{2} \frac{\xi - 2 \rho}{\xi \left( 1 - \rho^2 \right)}, \sqrt{1-\rho^2} (\theta+V_0)\xi^{-1}, - \rho \left(\theta + V_0\right) \xi^{-1}, \frac{1}{1-\rho^2}  \right) ,  &\mbox{if } H=-1/2, \\
\left( \frac{1}{2(1-\rho^2)}, - \frac{1}{2(1-\rho^2)}, \sqrt{1-\rho^2} \theta \xi^{-1}, - \rho \theta \xi^{-1}, \frac{1}{1-\rho^2} \right), \quad &\mbox{if } H<-1/2,
\end{cases}
\end{equation}
where $\theta$, $S_0$, $\xi$ and $V_0$ are the same from \eqref{eq:reversionaryHestonInstantaneousSpot}-\eqref{eq:reversionaryHestonInstantaneousVol}. Furthermore, the process $\bar V^{\epsilon}$ converges weakly towards $Y$ on the space $(D, SM_1)$, as $\epsilon \to 0$.
\end{theorem}
\begin{proof}  The  convergence of the characteristic functional in \eqref{eq:convchar} is established in Section~\ref{ss:proof_convergence_functional} (Lemmas~\ref{L:scaledRiccatiSolutionGoesToZero} and \ref{L:phieps}). This implies the convergence of the finite-dimensional  distributions of 
 $(\log S^{\epsilon},\bar V^{\epsilon})$ as detailed in Section~\ref{ss:rev_hreston_finite_dimensional_laws}. Finally, the weak convergence of $\bar V^{\epsilon}$ on $(D, SM_1)$ is proved in Section~\ref{ss:weak_convergence_V_bar}.
\end{proof}

\begin{remark}[An interesting interpretation]
The  interpretation of the convergence results becomes even more interesting when combined with Section~\ref{hyperRoughToReversionary}. In Section~\ref{hyperRoughToReversionary}, for $H>-1/2$,  the reversionary Heston model $(\log S^\epsilon, \bar V^{\epsilon})$ is constructed as a proxy of rough and hyper-rough Heston models.  Theorem~\ref{T:charfunlimit} shows that the limiting regime for $H>-1/2$ is a (degenerate) Black-Scholes regime, cf.Remark \ref{R:Norma}, whereas,  for $H\leq -1/2$ one obtains the convergence of the reversionary regimes towards (non-degenerate) jump processes with distinct regimes between $H=-1/2$ and $H<-1/2$, see Corollary \ref{Cor:revHestonMarginals} below. This suggests that jump models and (hyper-)rough volatility models are complementary, and do not overlap. For $H>-1/2$ the reversionary model can be interpreted as an engineering proxy of rough and hyper-rough volatility models, while for $H\leq -1/2$ it can be interpreted as {an approximation} of jump models for small enough $\epsilon$.  Asymptotically, jump models actually start at $H=-1/2$ (and below) in the Reversionary Heston model, the very first value of the Hurst index for which hyper-rough volatility models can no-longer be defined.
\end{remark}

\noindent In Figures \ref{F:mechkovimpli} and \ref{F:mechkovskew}, we plot respectively the convergence of the smiles and the skew of the reversionary Heston model $(\log S^{\epsilon}, \bar V^{\epsilon})$ for the case $H=-1/2$ towards the Normal Inverse Gaussian model. The volatility surface is obtained by applying Fourier inversion formulas on  the corresponding characteristic functions.  Similar to Figures~\ref{fig:smilesRoughHestonRevHeston} and \ref{fig:ATMskewsRoughHestonRevHeston}, the graphs show that the fast parametrizations introduced in the Heston model  are able to reproduce very steep skews for the implied volatility surface. 

\begin{figure}[H]
\begin{center}
\includegraphics[width=6 in,angle=0]{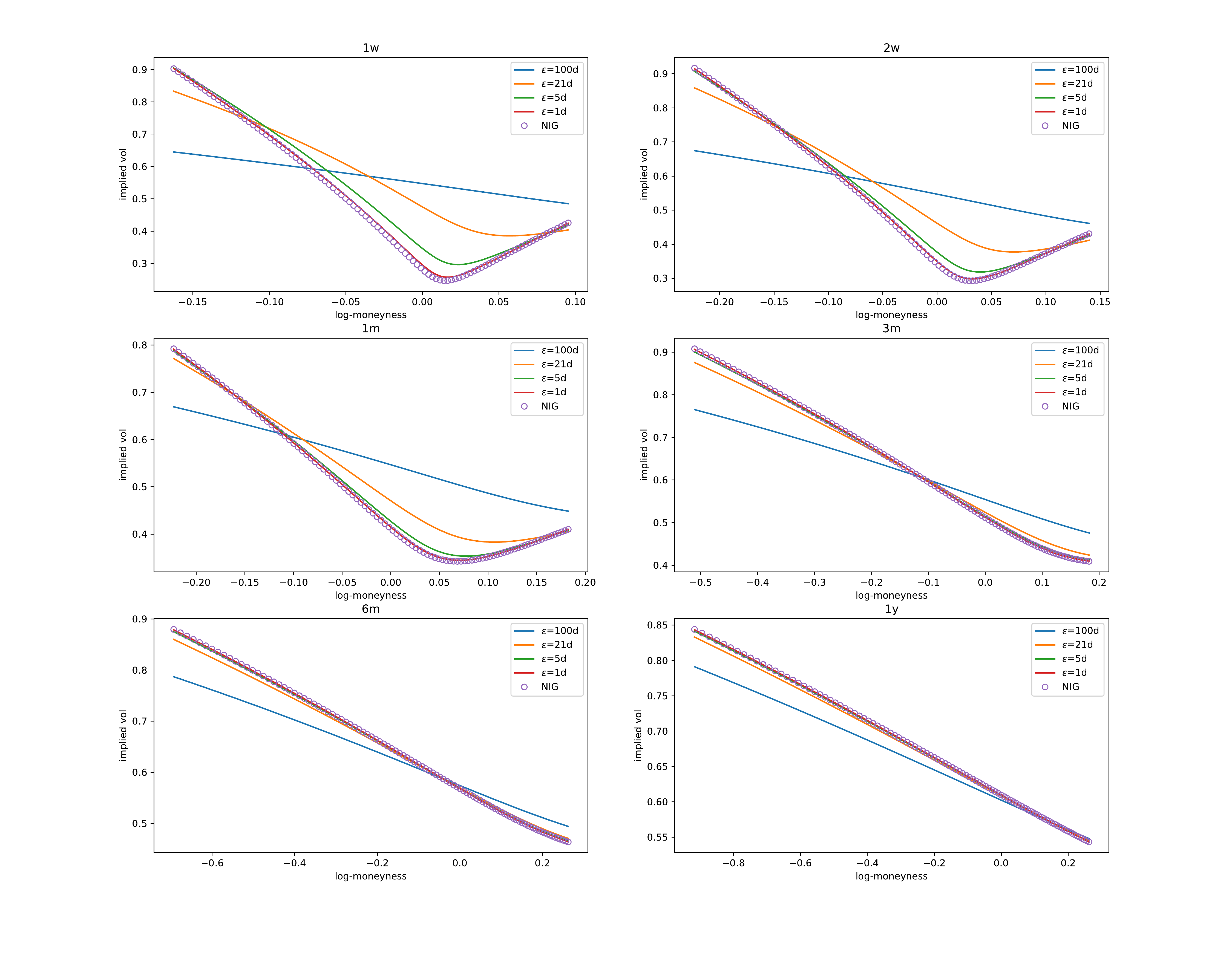}
\vspace*{-0.1in}
\caption{ Smiles of reversionary Heston and its asymptotic NIG law in the regime $H=-0.5$ for different maturities from one week to one year. Parameters are: $S_0=100, \; \rho=-0.7, \; \theta=0.3, \; \xi=0.8, \; V_0=0.3$, and the reversionary time-scale is varied from one hundred days to one day.}
\label{F:mechkovimpli}
\end{center}
\end{figure}

\begin{figure}[H]
\begin{center}
\includegraphics[scale=0.4]{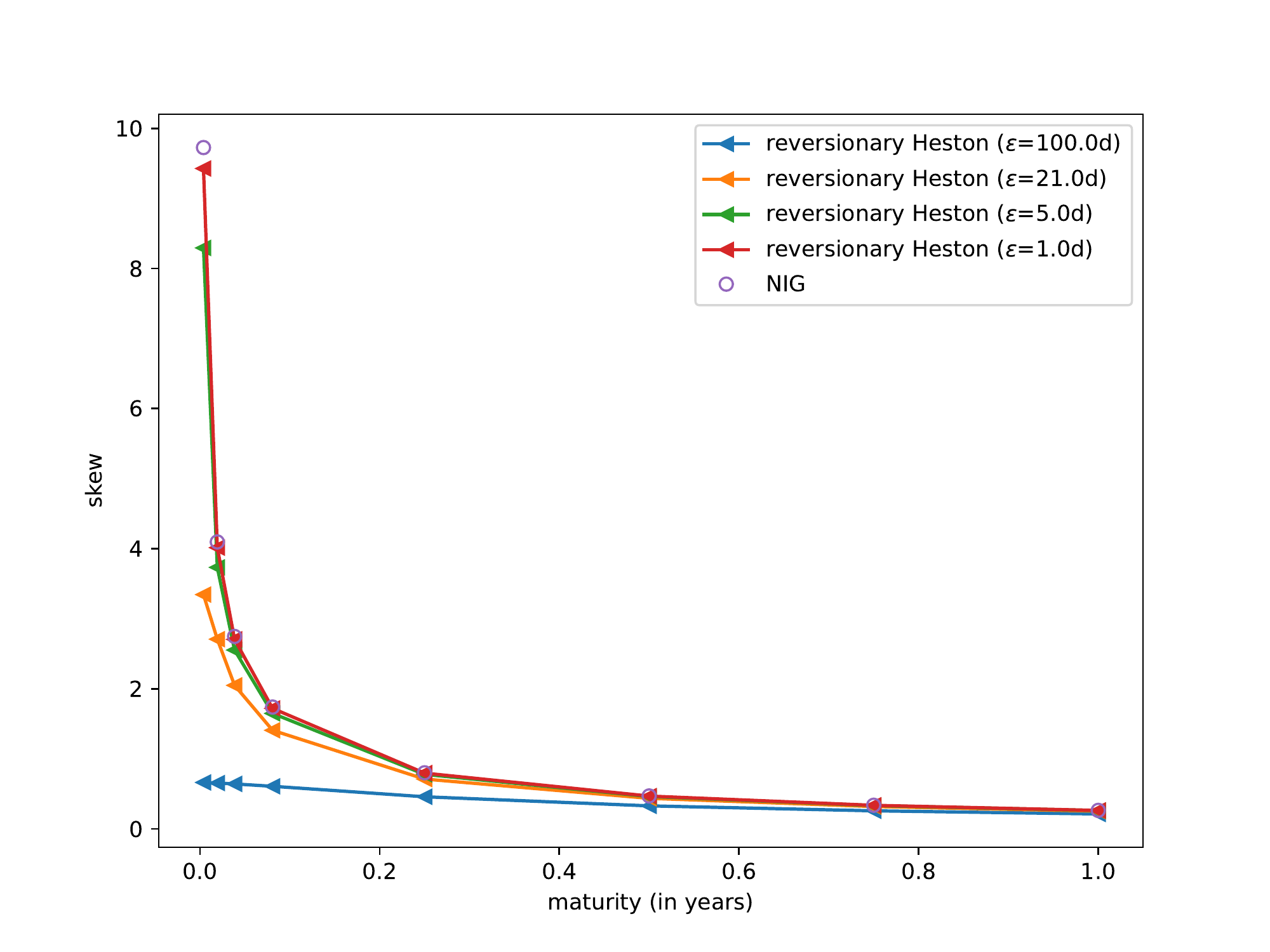}
\caption{At-the-money skew of reversionary Heston and its asymptotic NIG law in the regime $H=-0.5$. Parameters are: $S_0=100, \; \rho=-0.7, \; \theta=0.3, \; \xi=0.8, \; V_0=0.3$ and the reversionary time-scale is varied from one hundred days to one day.}\label{F:mechkovskew}
\end{center}
\end{figure}

\noindent In the case $\left(f,g\right) = \left( iu, iv \right)$, with $u, v \in \mathbb{R}$, the asymptotic marginals of reversionary Heston expressed in Corollary \ref{Cor:revHestonMarginals} below are obtained as a direct consequence of the convergence Theorem \ref{T:charfunlimit}. 

\begin{corollary}[Explicit asymptotic marginals] \label{Cor:revHestonMarginals}
{Based on the value of $H \in \mathbb{R}$, the pair of normalized log price and integrated variance} $\left(\log \frac{S^\epsilon_T}{S_0}, \bar{V}^\epsilon_T \right)$ has distinct asymptotic marginals as the reversionary time-scale $\epsilon$ goes to zero given by:
\begin{enumerate}
    \item $H>-1/2$, i.e.~Black Scholes-type asymptotic regime (BS regime)
    
    \begin{equation} \label{eq:CFBSregime}
    \mathbb E \left[ iu\log \frac{ S^\epsilon_T}{S_0} + iv\bar{V}^\epsilon_T \right] \underset{\epsilon \to 0}{\longrightarrow} \exp \left\{ - \frac{V_0}{2} \left( u^2 - 2i \left( v - \frac{u}{2} \right) \right) T \right\}.
    \end{equation}
    
    \item $H=-1/2$, i.e.~Normal Inverse Gaussian-type asymptotic regime (NIG regime)
    
    \begin{equation} \label{eq:CFNIGregime}
    \mathbb E \left[ iu\log  \frac{ S^\epsilon_T}{S_0}+ iv\bar{V}^\epsilon_T \right] \underset{\epsilon \to 0}{\longrightarrow} \exp \left\{ \left( \theta + V_0 \right) \xi^{-2} \left( 1 - i \rho \xi u -\sqrt{\left( 1 - i \rho \xi u \right)^2 - 2 \xi^2 \left( i v - \frac{u^2 + iu}{2} \right)} \right) T \right\}.
    \end{equation}
    
    \item $H<-1/2$, i.e.~Normal Lévy-type asymptotic regime (NL regime)
    
    \begin{equation} \label{eq:CFNormalLevy}
   \mathbb E \left[ iu\log \frac{ S^\epsilon_T}{S_0} + iv\bar{V}^\epsilon_T \right] \underset{\epsilon \to 0}{\longrightarrow} \exp \left\{ - \theta \xi^{-1} \left(i \rho u + \sqrt{\left( 1 - \rho^2 \right) u^2 - 2i \left( v - \frac{u}{2}\right)} \right) T \right\}.
    \end{equation}
    
\end{enumerate}
\end{corollary}

\noindent In Figure \ref{F:cf_convergence_all_regimes}, we illustrate numerically the convergence of the characteristic function in all three regimes.

\begin{figure}[H]
\begin{center}
\includegraphics[width=6 in,angle=0]{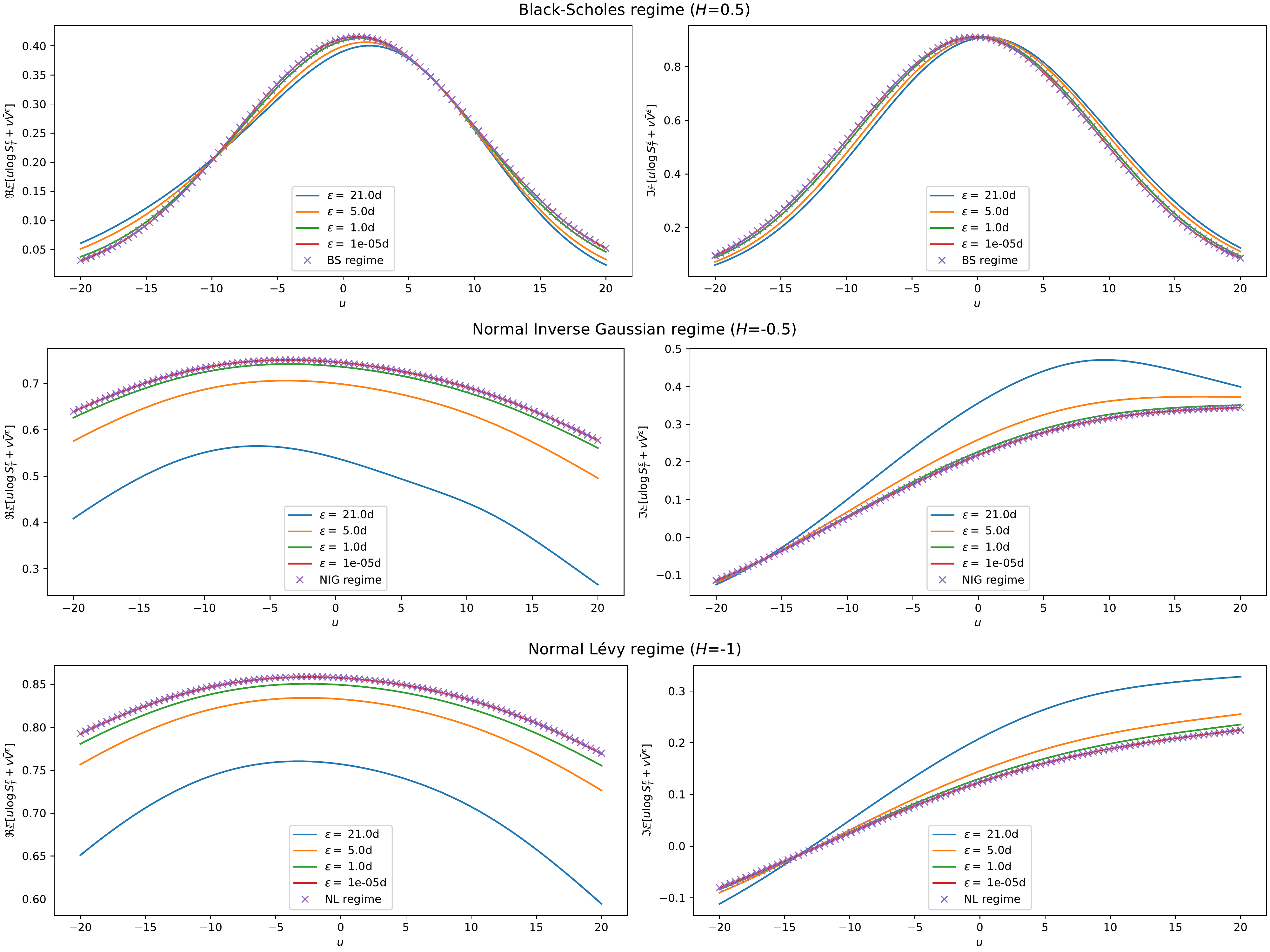}
\vspace*{-0.1in}
\caption{Convergence of reversionary Heston's joint characteristic function \eqref{eq:CHreversionaryHeston} in all three regimes whose asymptotic joint characteristic functions are given respectively in \eqref{eq:CFBSregime} for the Black-Scholes regime, in \eqref{eq:CFNIGregime} for the Normal Inverse Gaussian regime and in \eqref{eq:CFNormalLevy} for the Normal-Lévy regime. $v$ is fixed at $100$ and model parameters are $\rho=-0.7, \; \theta=0.3, \; \xi=0.8, \; V_0=0.3$ and the reversionary time-scale is varied from one $21$ days to $1e-5$ day.}
\label{F:cf_convergence_all_regimes}
\end{center}
\end{figure}

\noindent The rest of the section is dedicated to the proof of Theorem~\ref{T:charfunlimit}. 

\subsection{Convergence of the joint characteristic functional} \label{ss:proof_convergence_functional}
    
In this section we prove the convergence of the characteristic functional of $(\log {\frac{S^{\epsilon}}{S_0}}, \bar V^{\epsilon})$ as $\epsilon$ goes to $0$ stated in Theorem~\ref{T:charfunlimit}. For this, we fix 
$f,g:[0,T]\to \mathbb C$ bounded and measurable such that $\Re f =\Re g = 0$. We note that \eqref{eq:condfg} is trivially satisfied so that an application of Theorem~\ref{T:charfuneps}, with $t=0$, yields that 
\begin{equation}\label{eq:chart0}
          \mathbb{E} \left[ \exp\left(\int_0^T f(T-s) d\log S^\epsilon_s + \int_0^T g(T-s) d \bar{V}_s^\epsilon\right) \right] = \exp\left( \phi_{\epsilon}(T) + \epsilon^{\frac 1 2  - H }\psi_{\epsilon}(T) V_0 \right),
\end{equation}
    where $(\phi_{\epsilon},\psi_{\epsilon})$ solve \eqref{eq:Ric1}-\eqref{eq:Ric2}.
  We start by showing that the second term in the exponential  $\epsilon^{1/2-H}\psi_{\epsilon}$ goes to $0$ for any value of $H\in \R$ in the following lemma. 
  \begin{lemma} \label{L:scaledRiccatiSolutionGoesToZero}
For $\epsilon>0$, let $\psi_\epsilon$ be a solution  the time-dependent Riccati ODE \eqref{eq:Ric2} such that $\Re(\psi^{\epsilon})\leq 0$ with   $f, g:[0,T]\to \mathbb C$  bounded and measurable functions such that $\Re f = \Re g = 0$. Then, 
\begin{equation}
        \left| \psi_\epsilon(t) \right| \le C \epsilon^{H+\frac{1}{2}} \left( 1 - e^{- \epsilon^{-1}t} \right), \quad t\leq T, 
\end{equation}
    for some constant $C$ independent of $\epsilon$. In particular, we have the uniform convergence
    \begin{equation}
    \lim_{\epsilon \to 0} \sup_{t\leq T} \epsilon^{1/2-H}|\psi_\epsilon(t)| = 0, \quad H \in \R.
    \end{equation}
\end{lemma}

\begin{proof} 
 The variation of constants applied to the differential equation \eqref{eq:reformulation_ODE_psi_eps} yields:
\begin{equation} \notag
    \psi_\epsilon(t) = \int_ 0^t c_\epsilon(s)e^{\int_s^t \left( a_\epsilon \psi_\epsilon(u) + b_\epsilon(u) \right) du} ds,
\end{equation}
with $a_{\epsilon}, b_{\epsilon}, c_{\epsilon}$ defined as in \eqref{eq:abc}.  Given that $f$ and $g$ are both bounded on $[0,T]$, we fix $C \ge 0$ such that $\left| c_\epsilon(u) \right| \le C \epsilon^{H-\frac{1}{2}}, u \le T$. Note that $\Re b_\epsilon = - \epsilon^{-1}$, having also $ \Re\left( \psi_\epsilon\right) \le 0$ and $a_\epsilon > 0$, we get consequently by linearity when taking the real part

\begin{align}
    \left| \psi_\epsilon(t) \right| & \le C \epsilon^{H-\frac{1}{2}} \int_0^t e^{a_{\epsilon} \int_s^t \Re\left( \psi_\epsilon (u) \right) du} e^{\Re b_\epsilon (t-s)} ds\\
    & \le C \epsilon^{H-\frac{1}{2}} \int_0^t e^{- \epsilon^{-1}(t-s)} ds
\end{align}
which yields the desired upper bound on the solution when integrating explicitly the exponential. The convergence result follows immediately.
\end{proof}

\noindent The first term $\phi_{\epsilon}$, however,  yields  different limits based on the value of $H$. Consequently, we will study in Lemma~\ref{L:phieps} the convergence of the following quantity
\begin{equation*}
    \phi_\epsilon (t) = \int_0^t \left( \theta + \epsilon^{-H-\frac{1}{2}} V_0 \right) \psi_\epsilon(s) ds,
\end{equation*}
 for different regimes of $H$. 

 \begin{lemma} \label{L:phieps}
    We have the convergence
 \begin{align}
     \lim_{\epsilon\to 0} \phi_{\epsilon}(T) = \phi_0(T),
 \end{align}
 where  $\phi_0(T)$ is given by \eqref{eq:phi1}.
 \end{lemma}
 
 \begin{proof}
\textbf{Case $H>-1/2$.}
    In this case, the solution $\psi_\epsilon$ converges uniformly to zero from the upper bound given in Lemma \ref{L:scaledRiccatiSolutionGoesToZero} which, combined with the expression of $F$ in \eqref{eq:varconstantF},  yields 
\begin{align}
    \lim_{\epsilon\to 0} F(s,\psi_{\epsilon}(s))  = F(s,0) = g(s) + \frac{f^2(s) - f(s)}{2}, \quad s\leq T. 
\end{align}    
    Furthermore, integrating the variation of constants expression \eqref{eq:varconstantpsi} leads to
    \begin{align*}
        \int_0^t \epsilon^{-H-1/2} \psi_\epsilon(u) du & = \int_0^t \left\{ \int_0^u K_\epsilon(u-s) F(s,\psi_{\epsilon}(s)) ds \right\} du\\
        & = \int_0^t \left\{ \int_s^t K_\epsilon(u-s) du \right\} F(s,\psi_{\epsilon}(s)) ds\\
        & = \int_0^t \left( 1 - e^{-\epsilon^{-1}(t-s)} \right) F(s,\psi_{\epsilon}(s)) ds,
    \end{align*}
where we used Fubini for the second equality as the integrated quantity is bounded and measurable. Now, given the function $s \mapsto \left( 1 - e^{-\epsilon^{-1}(t-s)} \right)F(s,\psi_\epsilon(s))$ is uniformly bounded in $\epsilon$ by a constant on $[0,T]$, and that it converges pointwise to $s \mapsto F(0,s)$, then an application of  Lebesgue's Dominated Convergence Theorem yields
\begin{equation*}
    \int_0^t \epsilon^{-H-1/2} \psi_\epsilon(u) du \underset{\epsilon \to 0}{\longrightarrow} \int_0^t \left( g(s) + \frac{f^2(s) - f(s)}{2} \right)ds,
\end{equation*}
hence the resulting convergence 
\begin{equation*}
    \phi_\epsilon (t) \underset{\epsilon \to 0}{\longrightarrow} V_0 \int_0^t \left( g(s) + \frac{f^2(s) - f(s)}{2} \right) ds.
\end{equation*}

\noindent \textbf{Case $H=-1/2$.}
Now fix $\epsilon>0$ and using the second equation in \eqref{eq:psi0}, i.e.~$- F(t,\psi_{0}(t)) + \psi_0(t) = 0$, observe that 
     \begin{align}
        \psi_\epsilon'(t) & = - \epsilon^{-1} \psi_\epsilon(t) + \epsilon^{-1} \left( F(t,\psi_{\epsilon}(t)) - F(t,\psi_{0}(t)) + \psi_0(t) \right)\\
         &= - \epsilon^{-1} \left( \psi_\epsilon(t) - \psi_0(t) \right) + \epsilon^{-1} \beta_\epsilon(t) \left( \psi_{\epsilon}(t) - \psi_{0}(t) \right), \label{eq:psi_diff_expression}\\
         \text{with }\beta_\epsilon(t) & := \frac{\xi^2}{2} \left( \psi_{\epsilon}(t) + \psi_{0}(t) \right) + \rho \xi f(T-t). \label{eq:betavar}
     \end{align}

     \noindent Since $\psi_0$ has bounded variations by assumption, the complex-valued Riemann-Stieltjes integral on continuous functions $h$ against $\psi_0$ is well-defined, see Theorem A.1 from \cite{Montgomery_2007}. Define 
        $\Delta_{\epsilon}(t):=(\psi_{\epsilon}(t)-\psi_0(t))e^{t/\epsilon}$. Then, it follows that 
    \begin{align}
        d\Delta_{\epsilon}(t) &= (\psi'_{\epsilon}(t)dt-d\psi_0(t))e^{t/\epsilon} +  \epsilon^{-1}(\psi_{\epsilon}(t)-\psi_0(t))e^{t/\epsilon}dt\\
        &=  \epsilon^{-1} \beta_\epsilon(t) \Delta_{\epsilon}(t) - e^{t/\epsilon}d\psi_0(t),
    \end{align}
    where we used \eqref{eq:psi_diff_expression} to get the second equality, and applying the variation of constants formula leads to 
    \begin{align*}
        \Delta_{\epsilon}(t) = e^{\epsilon^{-1} \int_0^t\beta_\epsilon(r)dr }  
        \Delta_{\epsilon}(0)  - \int_0^t e^{\epsilon^{-1} \int_s^t\beta_\epsilon(r)dr }e^{\epsilon^{-1}s} d\psi_0(s),
    \end{align*}
    so that, recalling that $\psi_{\epsilon}(0)=0$,
    \begin{align}
        \psi_{\epsilon}(t) - \psi_0(t) &=  -e^{\epsilon^{-1} \int_0^t\beta_\epsilon(r)dr }e^{-\epsilon^{-1}t}  
        \psi_{0}(0)  - \int_0^t e^{\epsilon^{-1} \int_s^t\beta_\epsilon(r)dr}e^{-\epsilon^{-1}(t-s)} d\psi_0(s) \\
        & =: \bold I_{\epsilon}(t) + \bold {II}_{\epsilon}(t) .
     \end{align}
We now prove successively that $|\int_0^T \bold I_{\epsilon}(t)dt| \to 0$ and $|\int_0^T \bold {II}_{\epsilon}(t)dt| \to 0$.

\begin{itemize}
    \item \noindent Given that $\Re \beta_{\epsilon} \leq 0$, then

\begin{equation}
    \left| e^{\epsilon^{H-1/2} \int_0^t\beta_\epsilon(r)dr }e^{-\epsilon^{-1}t}  
 \psi_{0}(0) \right| \leq e^{-\epsilon^{-1}t} \left| \psi_{0}(0) \right| \to 0, \quad t \in (0,T),
\end{equation}
so that $\bold I_{\epsilon}(t)$ converges pointwise to $0$ on $(0,T)$  and is dominated by an integrable function, hence $|\int_0^T \bold I_{\epsilon}(t)dt| \to 0$ by Lebesgue's dominated convergence theorem.\\

    \item Regarding the second term, we have
    \begin{align*}
        \left| \int_0^T \int_0^t e^{\epsilon^{H-1/2} \int_s^t\beta_\epsilon(r)dr}e^{-\epsilon^{-1}(t-s)} d\psi_0(s) dt\right| & \leq \int_0^T \int_0^t e^{-\epsilon^{-1}(t-s)} \left| d\psi_0(s) \right| dt\\
        & = \int_0^T \int_s^T e^{-\epsilon^{-1}(t-s)} dt \left| d\psi_0(s) \right|\\
        & = \int_0^T \epsilon \left( 1 - e^{-\epsilon^{-1} \left( T-s\right)}\right) \left| d\psi_0(s) \right|
    \end{align*}
    where the inequality comes from the positivity of the first integration, an application of Theorem A.4\footnote{Suppose that $g$ has bounded variation, and put $g^*(x) := \underset{\sigma_N([0,x])}{\sup} \sum_{i=1}^N \left| g(x_n)-g(x_{n-1}) \right| $, with $\sigma_N([0,x])$ a $N$-points subdivision of $[0,x]$. Then for $h$ continuous, we have $$\left| \int_0^T h(x) dg(x) \right| \leq \int_0^T \left| h(x) \right| \left| dg(x) \right|,$$ provided that both integrals exist, and where we used the notation $\left| dg(.) \right| := dg^*(.)$.} from \cite{Montgomery_2007} and using again that $\Re \beta_{\epsilon} \leq 0$, with the positive measure on the right-hand side being the total variation measure defined as in Theorem 6.2 from \cite{Rudin_1966}, and we used Fubini-Lebesgue to get the first equality. Noting the point-wise convergence of the function $f_{\epsilon} : s \mapsto\epsilon \left( 1 - e^{-\epsilon^{-1} \left( T-s\right)}\right)$ to zero and its uniform boundedness in $\epsilon$ by a constant on $[0,T]$, the dominated convergence theorem applied to the total variation measure proves the result.
\end{itemize}
Thus, we obtained
\begin{equation*}
    \int_0^T \psi_{\epsilon}(t) dt \underset{\epsilon \to 0}{\longrightarrow} \int_0^T \psi_0(t) dt,
\end{equation*}
where $\psi_0$ satisfies the second equation in \eqref{eq:psi0}, hence we get
\begin{equation*}
    \phi_\epsilon (T) \underset{\epsilon \to 0}{\longrightarrow} \left( \theta + V_0 \right) \int_0^T \psi_0(s) ds,
\end{equation*}
which is the desired convergence.\\

\noindent \textbf{Case $H<-1/2$.} Define $\psi_0$ as the root with non-positive real part from the third equation in \eqref{eq:psi0}, recall Lemma \ref{L:lemmaExistenceRootWithNegRealPart}. Fixing again $\epsilon>0$, we have
\begin{align*}
     \psi_\epsilon'(t) & = - \epsilon^{-1} \psi_\epsilon(t) + \epsilon^{H-1/2} \left( F(t,\psi_{\epsilon}(t)) - F(t,\psi_{0}(t)) \right)\\
     &= - \epsilon^{-1} \psi_\epsilon(t) + \epsilon^{H-1/2} \beta_\epsilon(t) \left( \psi_{\epsilon}(t) - \psi_{0}(t) \right),
\end{align*}
\text{with} $\beta_\epsilon$ given by \eqref{eq:betavar}. Similarly the case $H<-1/2$, computing the differential of $\Delta_{\epsilon}(t):=(\psi_{\epsilon}(t)-\psi_0(t))e^{t/\epsilon}$ and applying the variation of constants formula leads to
\begin{align*}
    \Delta_{\epsilon}(t) = e^{\epsilon^{H-1/2} \int_0^t\beta_\epsilon(r)dr }  
    \Delta_{\epsilon}(0)  - \int_0^t e^{\epsilon^{H-1/2} \int_s^t\beta_\epsilon(r)dr }e^{s/\epsilon} \left(d\psi_0(s) + \epsilon^{-1} \psi_0(s) ds \right),
\end{align*}
 so that 
\begin{align}
     \psi_{\epsilon}(t) - \psi_0(t) &=  -e^{\epsilon^{H-1/2} \int_0^t\beta_\epsilon(r)dr }e^{-\epsilon^{-1}t}  
 \psi_{0}(0)  - \int_0^t e^{\epsilon^{H-1/2} \int_s^t\beta_\epsilon(r)dr}e^{-\epsilon^{-1}(t-s)} \left(d\psi_0(s) + \epsilon^{-1} \psi_0(s) ds \right)\\
 &= -e^{\epsilon^{H-1/2} \int_0^t\beta_\epsilon(r)dr }e^{-\epsilon^{-1}t}  
 \psi_{0}(0)  - \int_0^t e^{\epsilon^{H-1/2} \int_s^t\beta_\epsilon(r)dr}e^{-\epsilon^{-1}(t-s)} d\psi_0(s) \\
 &\quad - \int_0^t e^{\epsilon^{H-1/2} \int_s^t\beta_\epsilon(r)dr}e^{-\epsilon^{-1}(t-s)} 
 \epsilon^{-1} \psi_0(s) ds\\
& =: \bold I_{\epsilon}(t) + \bold {II}_{\epsilon}(t)  + \bold {III}_{\epsilon} (t).
 \end{align}
We already have from the previous case $H=-1/2$ that, as $\epsilon \to 0$, both integrals $\int_0^T \bold I_{\epsilon}(t) dt$, $\int_0^T \bold {II}_{\epsilon}(t) dt$ converge to $0$, all that remains to show consequently is that $\int_0^T \bold {III}_{\epsilon}(t) dt$ converges to $0$ too.
{\begin{itemize}
    \item A finer upper bound on $\Re \beta_{\epsilon}$ is required to deal with this third term. We already know from Theorem \ref{existenceTheorem} that $\Re \psi_{\epsilon} \leq 0$, and by definition of $\psi_0$, we get the following bound
    \begin{equation*}
        \Re \beta_{\epsilon} = \frac{\xi^2}{2} \left( \Re \psi_{\epsilon} + \Re \psi_{0} \right) \leq \frac{\xi^2}{2} \Re \psi_{0} \leq 0.
    \end{equation*}
    Set
    \begin{equation*}
        E := \left\{ s \in [0,T], \left( f(s), g(s) \right) \neq \left( 0, 0 \right) \right\},
    \end{equation*}
    and from \eqref{eq:psi00}, we know that $\psi_{0}=0$ on $[0,T] \backslash E$ while lemma \ref{L:boundedness_gamma_ratio} yields $\Re \psi_0 \neq 0$ on $E$ so that we can  bound $\bold {III}_{\epsilon}$ as follows
    \begin{align*}
        \left| \bold {III}_{\epsilon} (t) \right| & \leq \int_{[0,t] \cap E} \epsilon^{-1} |\psi_0(s)| e^{\epsilon^{H-1/2} \int_s^t\Re \beta_\epsilon(r)dr} e^{-\epsilon^{-1}(t-s)} ds \\
        & \leq \int_{[0,t] \cap E} \epsilon^{-1} \frac{|\psi_0(s)|}{-\epsilon^{H-1/2}\frac{\xi^2}{2} \Re \psi_0(s)} \left( -\epsilon^{H-1/2}\frac{\xi^2}{2}\Re \psi_0(s) e^{\int_s^t \epsilon^{H-1/2} \frac{\xi^2}{2} \Re \psi_0(r) dr} \right) ds \\ 
        & \leq \int_{[0,t] \cap E} \epsilon^{-H-1/2} 2 \xi^{-2} \frac{|\psi_0|}{- \Re \psi_0} \left( -\epsilon^{H-1/2}\frac{\xi^2}{2}\Re \psi_0(s) e^{\int_s^t \epsilon^{H-1/2} \frac{\xi^2}{2} \Re \psi_0(r) dr} \right) ds,
    \end{align*}
    and an application of lemma \ref{L:boundedness_gamma_ratio} yields the existence of a finite positive constant $C>0$ such that
    \begin{equation*}
        \frac{|\psi_0|}{- \Re \psi_0} \leq C, \quad 
 \forall s \in E,
    \end{equation*}
    so that
    \begin{align*}
        \left| \bold {III}_{\epsilon} (t)\right| & \leq C \epsilon^{-H-1/2} \int_{[0,t] \cap E} \left( -\epsilon^{H-1/2}\frac{\xi^2}{2}\Re \psi_0(s) e^{\int_s^t \epsilon^{H-1/2} \frac{\xi^2}{2} \Re \psi_0(r) dr} \right) ds\\
        & \leq C \epsilon^{-H-1/2} \left( 1 - e^{\int_0^t \epsilon^{H-1/2} \frac{\xi^2}{2} \Re \psi_0(r) dr} \right) \underset{\epsilon \to 0}{\longrightarrow} 0,
    \end{align*}
    since $-H-1/2>0$ so that $\epsilon^{-H-1/2} \to 0$ as $\epsilon \to 0$. Thus $\bold {III}_{\epsilon}$ is dominated by a finite constant $C$ independent of $\epsilon$ (which is integrable on $[0,T]$) and Lebesgue's dominated convergence theorem yields that $\int_0^T \bold {III}_{\epsilon}(t) dt$ converges to $0$.
\end{itemize}
Consequently, we obtained
\begin{equation*}
    \int_0^T \psi_{\epsilon}(t) dt \underset{\epsilon \to 0}{\longrightarrow} \int_0^T \psi_0(t) dt,
\end{equation*}
which then yields
\begin{equation*}
    \epsilon^{-H-1/2} \int_0^T \psi_{\epsilon}(t) dt \underset{\epsilon \to 0}{\longrightarrow} 0,
\end{equation*}
and finally we get
\begin{equation*}
    \phi_\epsilon (T) \underset{\epsilon \to 0}{\longrightarrow} \theta \int_0^T \psi_0(s) ds.
\end{equation*}
}
\end{proof}

\subsection{Convergence of the finite-dimensional distributions towards  NIG-IG} \label{ss:rev_hreston_finite_dimensional_laws}

In this section, we prove the  second part of Theorem~\ref{T:charfunlimit}, that is the convergence of the finite-dimensional distributions of $\left( {\log \frac{S^\epsilon}{S_0}}, \bar V^\epsilon \right)$ towards those of a NIG-IG process  $\left( X, Y \right)$ in the sense of Definition~\ref{D:NIGIG}  with parameters $(\alpha, \beta, \mu, \delta,\lambda)$ as in \eqref{eq:paramlimit} depending on the regime of $H$. Let $d \in \mathbb{N}^*$ and take $0 =: t_0 < t_1 < \dots < t_d \leq T$ to be $d$ distinct times of the time interval $\left[0, T\right]$ and $\left( u_k, v_k \right)_{k \in \left\{ 1, \dots, d \right\}} \in \left(\mathbb{R}^{2}\right)^d$. We will prove that 
\begin{equation}\label{eq:convtemp}
     \mathbb{E} \left[ \exp \left( i\sum_{k=1}^d u_k \log {\frac{S^\epsilon_{t_k}}{S_0}} + i\sum_{k=1}^d v_k \bar{V}_{t_k}^\epsilon \right) \right] \underset{\epsilon \to 0}{\longrightarrow} \; \mathbb{E} \left[ \exp \left( i\sum_{k=1}^d u_k X_{t_k} + i\sum_{k=1}^d v_k Y_{t_k} \right) \right].
\end{equation}
{First, we recover the finite-dimensional distributions of $\left(\log \frac{S^\epsilon_T}{S_0}, \bar{V}^\epsilon_T \right)$ }from \eqref{eq:convchar} by setting  the bounded and measurable functions $f$ and $g$ to be equal to
\begin{equation*}
    f(s) := i\sum_{k=1}^d \mathbb{1}_{\left[t_{k-1}, t_k\right)}(T-s) \sum_{k=j}^d u_j \quad \mbox{and} \quad
    g(s)  := i\sum_{k=1}^d \mathbb{1}_{\left[t_{k-1}, t_k\right)}(T-s) \sum_{k=j}^d v_j.
\end{equation*}
 Notice indeed that
\begin{align*}
   i \sum_{k=1}^d u_k \log { \frac{S^\epsilon_{t_k}}{S_0}} + i\sum_{k=1}^d v_k \bar{V}_{t_k}^\epsilon  &=  i\sum_{k=1}^d \left(\log { \frac{S^\epsilon_{t_k}}{S_0}}-\log {\frac{S^\epsilon_{t_{k-1}}}{S_0}}\right) \sum_{j=k}^d u_j + i\sum_{k=1}^d \left(\bar{V}_{t_k}^\epsilon-\bar{V}_{t_{k-1}}^\epsilon\right) \sum_{j=k}^d v_j \\
    &=  \int_0^T f(T-s) d\log {S^\epsilon_s} + \int_0^T g(T-s) d \bar{V}_s^\epsilon, 
    \end{align*} 
and that  the corresponding $\psi_0$ defined in \eqref{eq:psi0} has bounded variations (being piece-wise constant for this choice of $f$ and $g$), so that an application of the convergence of the characteristic functional in \eqref{eq:convchar} yields
\begin{align*}
    \mathbb{E} \left[ \exp \left( i\sum_{k=1}^d u_k \log {\frac{S^\epsilon_{t_k}}{S_0}} + i\sum_{k=1}^d v_k \bar{V}_{t_k}^\epsilon \right) \right] {\longrightarrow} \exp \left( \phi_0(T) \right),
\end{align*}
with
     \begin{align}
    \phi_0(T) 
    & = \begin{cases}
    V_0 \sum_{k=1}^d (t_k -  t_{k-1}) (i\bar v_{k}- \frac{\bar u_{k}^2 + i\bar u_{k}}{2}) , \quad &\mbox{if } H>-1/2,\\
    \left( \theta + V_0 \right) \xi^{-2} \sum_{k=1}^d(t_k -  t_{k-1}) \left( 1 - \rho \xi i\bar u_{k} + \sqrt{\left( 1 - \rho \xi i\bar u_{k}\right)^2 - 2 \xi^2 (i\bar v_{k}- \frac{\bar u_{k}^2+i\bar u_{k}}{2})} \right),  &\mbox{if } H=-1/2, \\
    -\theta \xi^{-1} \sum_{k=1}^d (t_k -  t_{k-1}) \left( \rho i\bar u_{k} + \sqrt{-\rho^2 \bar u_{k}^2 - 2 (i\bar v_{k}- \frac{\bar u_{k}^2+i\bar u_{k}}{2})} \right), \quad &\mbox{if } H<-1/2,
    \end{cases}
    \end{align}
where we defined $\bar u_{k} := \sum_{j=k}^d u_j$, $\bar v_{k} := \sum_{j=k}^d v_j$.\\

\noindent Second, we identify such $\phi_0(T)$ with the corresponding finite-dimensional distributions of the NIG-IG process $\left( X, Y \right)$ with parameters $(\alpha, \beta, \mu, \delta,\lambda)$ as in \eqref{eq:paramlimit} depending on the regime of $H$. We denote by   $\eta$ its Lévy exponent, recall  \eqref{eq:eta_definition}, and we write 
\begin{align}
    \mathbb{E} \left[ \exp \left( i\sum_{k=1}^d u_k X_{t_k} + i\sum_{k=1}^d v_k Y_{t_k} \right) \right] & = \mathbb{E} \left[ \exp \left( i\sum_{k=1}^d \left(X_{t_k}-X_{t_{k-1}}\right) \sum_{j=k}^d u_j + i\sum_{k=1}^d \left(Y_{t_k}-Y_{t_{k-1}}\right) \sum_{j=k}^d v_j\right) \right]\\
    & = \prod_{k=1}^d \mathbb{E} \left[ \exp \left( \left(X_{t_k}-X_{t_{k-1}}\right) i\sum_{j=k}^d u_j + \left(Y_{t_k}-Y_{t_{k-1}}\right) i\sum_{j=k}^d v_j\right) \right]\\
    & = \prod_{k=1}^d \mathbb{E} \left[ \exp \left( X_{t_k-t_{k-1}} i\sum_{j=k}^d u_j + Y_{t_k-t_{k-1}} i\sum_{j=k}^d v_j\right) \right]\\
    & = \exp \left( \sum_{k=1}^d \left( t_k - t_{k-1} \right) \eta \left( \bar u_k, \bar v_k \right) \right)
\end{align}
using respectively telescopic summation, the independence of increments, the fact that $\left( X_{t_2}, Y_{t_2} \right) - \left( X_{t_1}, Y_{t_1} \right) \overset{\textit{law}}{=} \left( X_{t_2-t_1}, Y_{t_2-t_1} \right)$ and the definition of the characteristic function of $\left( X_{t_{k}-t_{k-1}}, Y_{t_{k}-t_{k-1}} \right)$ for all $k \in 1, \dots, d$ to get the successive equalities. Using the parameters  $(\alpha, \beta, \mu, \delta,\lambda)$ as in \eqref{eq:paramlimit}  it is immediate to see that  
$$\sum_{k=1}^d \left( t_k - t_{k-1} \right) \eta \left( \bar u_k, \bar v_k \right) = \phi_0(T),$$
hence the desired convergence \eqref{eq:convtemp}.

\subsection{Weak-convergence of the integrated variance process for the $M_1$ topology} \label{ss:weak_convergence_V_bar}

In this section, we prove the weak convergence stated in Theorem~\ref{T:charfunlimit} of the integrated variance process $\bar{V}^\epsilon$ with sample paths in ${C}\left([0,T], \mathbb{R}_+\right)$ to the Lévy process $\bar{V}^0$ whose Lévy exponent is given by $v \mapsto \eta(0,v)$ with $\eta$ defined in \eqref{eq:eta_definition} and with sample paths in the real-valued càdlàg functional space $D$ endowed with Skorokhod's Strong M1 ($SM_1$) topology. There will be two subsections: first, in Section \ref{ss:SM1_topology_remainder} we recall briefly the definition of the Strong M1 ($SM_1$) topology as well as some associated convergence results, and then we prove the tightness of the integrated variance in Section~\ref{ss:conv}.

\subsubsection{Reminder on the $SM_1$ topology and conditions for convergence} \label{ss:SM1_topology_remainder}

\noindent We recall succinctly key definitions and convergence theorems for the strong $M_1$ topology, denoted $SM1$. We refer the reader to the key reference book \citet[Chapter 12]{Whitt_Ward_m1} for more details.  For $x \in D$, we define the thin graph of $x$ as
\begin{equation*}
    \Gamma_{x} := \left\{ (z,t) \in \mathbb{R} \times \left[ 0, T \right]: z \in \left[ x(t^{-}), x(t) \right] \right\},
\end{equation*}
where, for $t \in \left[ 0,T \right]$, $\left[ x(t^{-}), x(t) \right]$ denotes the standard segment $\left\{ \alpha x(t^{-}) + \left( 1 - \alpha \right)  x(t), 0 \leq \alpha \leq 1 \right\}$, which is different from a singleton at discontinuity points of the càdlàg sample trajectory $x$. We denote $Disc(x)$ the set of such instants. Define on $\Gamma_{x}$ the strong order relation as follows: $\left(z_{1},t_{1}\right) \leq \left(z_{2},t_{2}\right)$ if either $t_1<t_2$ or $t_1=t_2$ and $\left|x(t_{1}^-)-z_1\right| \leq \left|x(t_{1}^-)-z_2\right|$. Furthermore define a strong parametric representations of $x$ as a continuous non-decreasing (with respect to the previous order relation) function $\left( \hat x, \hat r \right)$ mapping $\left[ 0,1 \right]$ into $\Gamma_{x}$ such that $\left( \hat x(0), \hat r(0) \right) = \left( x(0),0 \right)$ and $\left( \hat x(T), \hat r(T) \right) = \left( x(T),T \right)$. We say the component $\hat r$ scales the time interval $\left[ 0,T \right]$ to $\left[ 0,1 \right]$ while $\hat x$ time-scales $x$, and we denote by $\Pi_{s}(x)$ the set of all strong parametric representations of $x$. Finally, the SM1 topology is the one induced by the metric $d_{s}$ defined as
\begin{equation}
    d_{s} \left( x_{1}, x_{2} \right) := \underset{\left( \hat x^{i}, \hat r^{i} \right) \in \Pi_{s}(x^{j}), j=1,2}{\inf} \left\{ \left| \hat x^{1} - \hat x^{2} \right| \vee \left| \hat r^{1} - \hat r^{2} \right| \right\}.
\end{equation}

\noindent Below we mention briefly five theorems in a row that eventually yields criteria to prove the desired convergence.

\begin{theorem}
    $D$ endowed with the $SM_1$ topology is \textit{Polish}, i.e.~metrizable as a complete separable metric space.
\end{theorem}

\begin{theorem} [Prohorov's theorem]
    Let $(S,m)$ be a metric space. If a subset $A$ in $\mathcal{P}(S)$ is tight, then it's relatively compact. On the other hand, if the subset $A$ is relatively compact and the topological space is Polish, then $A$ is tight.
\end{theorem}

\noindent The first two theorems ensure that, since $\left(D,SM_1 \right)$ is Polish, proving relative compactness of any family of probability measures on such space is sufficient to ensure the existence of a convergent sub-sequence. Finally the convergence of finite-dimensional laws allows to uniquely determine its limit as formulated in the following theorem.

\begin{theorem}[Criteria for convergence in distribution in $\left(D,SM_1\right)$] \label{T:CriteriaForConvergence}Let $\left( ( X_n)_{n\in \mathbb N}, X \right)$ be random functions defined on a common filtered probability space $\left( \Omega, \left( \mathcal{F}_t \right)_t, \mathbb{P}\right)$. Then, $X_n \Rightarrow X$ in $D$ for the $SM_1$ topology if the following conditions hold:
\begin{itemize}
        \item $(X_n)_{n\in \mathbb N}$ is tight with regards to the $SM_1$ topology.
        \item The finite dimensional distributions of $\left( X_n \right)$ converge to those of $X$ on $T_X$, where:
        \begin{equation}
            T_X := \left\{ t > 0, \mathbb{P} \left(t \in Disc(X) \right) = 0 \right\} \cup \{T\}.
        \end{equation}
        
    \end{itemize}
\end{theorem}

\noindent To conclude this section, we recall a characterization of tightness for a sequence of probability measures.

\begin{theorem}[Characterization of tightness]\label{T:characterization_tightness}
The sequence of probability measures $\left\{ \mathbb{P}_n \right\}_{n \ge 1}$ on $\left(D,SM_1\right)$ is tight if and only if:
\begin{itemize}
    \item[(i)] $\forall \bar \varepsilon>0, \exists c < \infty, \forall n \ge 1, \mathbb{P}_n \left( \left\{ x \in D: ||x|| > c \right\} \right) < \bar \varepsilon$
    \item[(ii)] $\forall \bar \varepsilon>0, \forall \eta > 0, \exists \delta>0, \forall n \ge 1, \mathbb{P}_n \left( \left\{ x \in D: w'(x,\delta) \ge \eta \right\} \right) < \bar \varepsilon$
\end{itemize}
Where we defined for $x \in D$, $t\in [0,T]$ and $\delta>0$:
\begin{align}
    ||x|| & := \underset{t \in [0,T]}{\text{sup}} |x_t|,\\
    w'(x,\delta) &:= w(x,\delta) \vee \bar{v}(x,0,\delta) \vee \bar{v}(x,T,\delta) ,\\ \notag
    w(x,\delta) &:= \underset{t \in [0,T]}{\text{sup}} w_S(x,t,\delta), \\ \notag
    w_S(x,t,\delta) &:= \underset{0 \vee t-\delta \le t_1 < t_2 < t_3 \le (t+\delta) \wedge T}{\text{sup}} \left| x(t_2)-\left[ x(t_1), x(t_3)\right] \right|, \\ \notag
    \bar{v}(x,t,\delta) &:= \underset{0 \vee t-\delta \le t_1 \le t_2 \le (t+\delta) \wedge T}{\text{sup}} \left| x(t_1)-x(t_2) \right|.
\end{align}
\end{theorem}

\subsubsection{Convergence of the integrated variance process}\label{ss:conv}

We already proved in Section \ref{ss:rev_hreston_finite_dimensional_laws} the convergence of finite dimensional distributions as $\epsilon$ goes to zero of $\bar{V}^\epsilon$ toward those of {either} the deterministic linear, or the Inverse Gaussian, or the Lévy process {denoted $Y$} depending respectively on the value of the {parameter $H$} with Lévy exponent $\eta(0,\cdot)$ from \eqref{eq:eta_definition} with the respective parameters given in Theorem \ref{T:charfunlimit}. Consequently, all that remains to prove is the tightness of the family of processes $\left( \bar{V}^\epsilon \right)_{\epsilon>0}$ for the $SM_1$ topology to get the desired convergence result as a direct consequence of Theorem \ref{T:CriteriaForConvergence}. We will apply the characterization Theorem \ref{T:characterization_tightness} of tightness in $SM_1$ to conclude, and more precisely, we will see that the criteria of tightness within the $SM_1$ topology simplifies greatly for almost surely non-decreasing and continuous stochastic processes in general.\\


\noindent Fix $\bar \varepsilon > 0$. Since, for all $\epsilon>0$, $\bar{V}^\epsilon$ is almost surely non-decreasing and non-negative, we have that, for all $\omega \in \Omega$
\begin{align*}
||\bar{V}^\epsilon(\omega)|| &= \bar V_T^{\epsilon}(\omega), \\
w'(\bar{V}^\epsilon(\omega), \delta) &= 0\vee \bar V^{\epsilon}_{\delta}(\omega) \vee \left( \bar V^{\epsilon}_{T}(\omega) - \bar V^{\epsilon}_{T-\delta}(\omega)\right), \quad \delta >0.
\end{align*}
This yields that, for a threshold $c>0$ big enough, the probability in condition $(i)$ on the measures $\mathbb{P}_\epsilon := \mathbb{P}_{\left(\bar{V}_.^\epsilon\right)^{-1}}$ reduces to
\begin{equation*}
    \mathbb{P}_\epsilon \left( \left\{ x \in D, ||x|| > c \right\} \right) = \mathbb{P} \left( \left\{ \omega: ||\bar{V}_.^\epsilon(\omega)||>c \right\} \right) = \mathbb{P} \left( \bar{V}_T^\epsilon>c \right) \le \sup_{\epsilon>0} \mathbb{P} \left( \bar{V}_T^\epsilon>c \right) < \bar \varepsilon,
\end{equation*}
where the last inequality is satisfied by tightness of the family $\left( \bar{V}_T^\epsilon \right)_{\epsilon>0}$ of random variables in $\R$ which is obtained as a direct consequence of Lévy's continuity theorem, recall that  $\left( \bar{V}_T^\epsilon \right)_{\epsilon>0}$ has been shown to converge in Section \ref{ss:proof_convergence_functional}. This yields $(i)$. In addition, regarding the second condition $(ii)$, set an arbitrary $\eta>0$, and take $\delta$ small enough such that
\begin{align*}
    \mathbb{P}_{\epsilon} \left( \left\{ x \in D: w'(x,\delta) \ge \eta \right\} \right) &  = \mathbb{P} \left( 0 \vee \left|\bar{V}_{\delta}^\epsilon \right| \vee \left|\bar{V}_T^\epsilon - \bar{V}_{T-\delta}^\epsilon\right| \ge \eta \right) \leq \sup_{\epsilon>0} \mathbb{P} \left( 0 \vee \left|\bar{V}_{\delta}^\epsilon \right| \vee \left|\bar{V}_T^\epsilon - \bar{V}_{T-\delta}^\epsilon\right| \ge \eta \right) < \bar \varepsilon,
\end{align*}
where the last but one inequality holds by tightness of the family $\left( \bar{V}_{\delta}^\epsilon, \bar{V}_{T-\delta}^\epsilon \right)_{\epsilon > 0}$ while the last one is justified by stochastic continuity of $\bar{V}^\epsilon$ for any $\epsilon > 0$ and of its limit $Y$.
\begin{remark}
    Since $\log S^\epsilon = - \frac{1}{2} \bar V^\epsilon + \rho W_{\bar V^\epsilon} + \sqrt{1-\rho^2} W^\perp_{\bar V^\epsilon}$, and composition is not continuous in $\left( D, M_1 \right)$ (see \cite[Section 3.5]{McCrickerd_2021}), we cannot expect the tightness of the log price within $SM_1$.
\end{remark}

\appendix 

\section{Some lemmas}

\begin{lemma} \label{L:lemmaExistenceRootWithNegRealPart}
    (Uniqueness of the complex root with a non-positive real part) 
    Take $\xi>0$. For all $f$, $g$ bounded and measurable with $\Re f = \Re g = 0$, $t \in \left[ 0,T \right]$ and $\rho \in \left[ -1,1 \right]$, both polynomials 
    
    \begin{align*}
        P(X) & := \frac{\xi^2}{2} X^2 - \left( 1 - \rho \xi f(t) \right) X + g(t) + \frac{f^2(t) - f(t)}{2}\\
        Q(X) & := \frac{\xi^2}{2} X^2 + \rho \xi f(t) X + g(t) + \frac{f^2(t) - f(t)}{2}
    \end{align*}
    admit exactly two roots with respective real parts of strict opposite signs if $\left( f(t), g(t) \right) \neq \left( 0, 0 \right)$, and if $\left( f(t), g(t) \right) = \left( 0, 0 \right)$, then the polynomial $P$ has roots $0$ and $\frac{2}{\xi^2}$, while $Q$ has $0$ as a double root.
\end{lemma}

\begin{proof}
    Let us detail the proof for $P$, similar arguments will apply to $Q$. By d'Alembert-Gauss theorem, the polynomial $P$ admits exactly two roots expressed as:
    
    \begin{equation} \notag
        \left\{ \xi^{-2} \left( 1 - \rho \xi f(t) \pm \sqrt{\left( 1 - \rho \xi f(t) \right)^2 - 2 \xi^2 \left( g(t) + \frac{f^2(t)-f(t)}{2} \right)} \right) \right\},
    \end{equation}
    where we take the principal square root in the expression above, i.e.with non-negative real-part. Consequently, the roots have real parts
    
    \begin{equation} \notag
        \left\{ \xi^{-2} \left( 1 \pm \Re \left(\sqrt{\left( 1 - \rho \xi f(t) \right)^2 - 2 \xi^2 \left( g(t) + \frac{f^2(t)-f(t)}{2} \right)}\right) \right) \right\},
    \end{equation}
    so that it remains to show $\left| \Re \left(\sqrt{\left( 1 - \rho \xi f(t) \right)^2 - 2 \xi^2 \left( g(t) + \frac{f^2(t)-f(t)}{2} \right)}\right) \right| > 1$.\\
    
    \noindent Denote $\delta=a+ib$, $a,b \in \mathbb{R}$ such that $\delta^2 = \left( 1 - \rho \xi f(t) \right)^2 - 2 \xi^2 \left( g(t) + \frac{f^2(t)-f(t)}{2} \right)$, then it follows that $a$ and $b$ satisfy
    
    \begin{equation} \label{inequalitiesRealImParts}
        \begin{cases}
            a^2 - b^2 & = 1 + \left( 1 - \rho^2 \right) \left( \xi \Im f(t) \right)^2, \\
            ab & = - \left( \rho \xi \Im f(t) + \xi^2 \left( \Im g(t) - \frac{\Im f(t)}{2} \right) \right).
        \end{cases}
    \end{equation}
    
    \noindent If $\rho \neq \pm 1$, then the result is immediate from the first inequality in \eqref{inequalitiesRealImParts}, while if $\rho = \pm 1$, then $|a| = \sqrt{1+b^2}$ and $b$ cannot be zero, otherwise $\delta^2 = a^2 = 1 - 2i\left( \rho \xi \Im f(t) + \xi^2 \left( \Im g(t) - \frac{\Im f(t)}{2} \right) \right)$ which cannot be the case, since $a \in \mathbb{R}$ and $\left(\Im f(t), \Im g(t)\right) \neq \left( 0, 0 \right)$.
\end{proof}

{\begin{lemma} \label{L:boundedness_gamma_ratio}
    Let $f$ and $g$ be bounded measurable functions such that $\Re f = \Re g = 0$. Then there exists a finite positive constant $C$ such that the ratio
    \begin{equation*}
        \gamma(s) := \frac{|\psi_0(s)|}{- \Re \psi_0(s)} \leq C, \forall s \in E,
    \end{equation*}
    where the set $E$ is given by
    \begin{equation*}
        E := \left\{ s \in [0,T], \left( f(s), g(s) \right) \neq \left( 0, 0 \right) \right\},
    \end{equation*}
    and $\psi_0$ is given in \eqref{eq:psi00} in the case $H<-1/2$.
\end{lemma}

\begin{proof}
    \noindent We start by explicitly computing the real and imaginary parts of $\psi_0$ in the case $H<-1/2$, whose expression is given anew by
    \begin{equation*}
        \psi_0(s) = - \xi^{-1} \left( \rho f(s) + \sqrt{f(s) \left( 1 - \left( 1- \rho^2 \right) f(s) \right) - 2 g(s)} \right), \quad s \in [0,T].
    \end{equation*}
    Set the real functions $a$ and $b$ such that, for any $s \in [0,T]$
    \begin{equation*}
        a(s) + i b(s) = \sqrt{f(s) \left( 1 - \left( 1- \rho^2 \right) f(s) \right) - 2 g(s)},
    \end{equation*}
    square the above equality, identify the real and imaginary parts, find $a$ unambiguously on $E$ (which imposes $a \neq 0$) as a root to a fourth-degree polynomial knowing the square roots in \eqref{eq:psi0} are principal (i.e.~with positive real parts), then deduce the associated $b$, such that
    \begin{align*}
        a & = \sqrt{\frac{1}{2} \left( \left( 1 - \rho^2 \right)\left( \Im f \right)^2 + \sqrt{\left( 1 - \rho^2 \right)^2\left( \Im f \right)^4 + \left( \Im f - 2 \Im g\right)^2} \right)}, \\
        b & = \frac{\frac{1}{2}\Im f - \Im g}{a}.
    \end{align*}

    \noindent Consequently, we get explicitly on $E$
    \begin{equation*}
        \psi_0 = \Re \psi_0 + i \Im \psi_0,
    \end{equation*}
    with
    \begin{align*}
        \Re \psi_0 & = -\xi^{-1} a,\\
         \Im \psi_0 & = -\xi^{-1}\left( \rho \Im f + \frac{\frac{1}{2}\Im f - \Im g}{a} \right).
    \end{align*}
    
    \noindent Rewrite $\gamma$ as
    \begin{equation*}
        \gamma = \sqrt{ 1 + \tilde \gamma^2}, \quad \tilde \gamma := \frac{\Im \psi_0}{\Re \psi_0},
    \end{equation*}
    and we can readily discard the case $\rho=0$ and $\Im f - 2 \Im g=0$, since $\Im \psi_0 = 0$ in that case and the ratio simplifies into $1$ which yields the result. Assume from now on that $\rho \neq 0$ or $\Im f - 2 \Im g \neq 0$. We can write
    \begin{equation*}
        \tilde \gamma = \bold{A}\left( \Im f, \Im g \right) + \bold{B}\left( \Im f, \Im g \right),
    \end{equation*}
    with
    \begin{align*}
        \bold{A}\left( \Im f, \Im g \right) & := \frac{\rho \Im f}{\sqrt{\frac{1}{2} \left( \left( 1 - \rho^2 \right)\left( \Im f \right)^2 + \sqrt{\left( 1 - \rho^2 \right)^2\left( \Im f \right)^4 + \left( \Im f - 2 \Im g\right)^2} \right)}},\\
        \bold{B}\left( \Im f, \Im g \right) & := \frac{\frac{1}{2}\Im f - \Im g}{\frac{1}{2} \left( \left( 1 - \rho^2 \right)\left( \Im f \right)^2 + \sqrt{\left( 1 - \rho^2 \right)^2\left( \Im f \right)^4 + \left( \Im f - 2 \Im g\right)^2} \right)},
    \end{align*}
    so that there are three remaining cases for which it is sufficient to show that both $\bold{A}\left( \Im f, \Im g \right)$ and $\bold{B}\left( \Im f, \Im g \right)$ are bounded to conclude the proof.
    \begin{itemize}
        \item Case $\rho \neq 0$ and $\Im f - 2 \Im g = 0$, then
        \begin{align*}
            \bold{A}\left( \Im f, \Im g \right) & = \frac{\rho}{\sqrt{1-\rho^2}},\\
            \bold{B}\left( \Im f, \Im g \right) & = 0,
        \end{align*}
        are both bounded, recall $\rho \in (-1,1)$.

        \item Case $\rho = 0$ and $\Im f - 2 \Im g \neq 0$, then
        \begin{align*}
            \bold{A}\left( \Im f, \Im g \right) & = 0,\\
            \bold{B}\left( \Im f, \Im g \right) & = \frac{\frac{1}{2}\Im f - \Im g}{\frac{1}{2} \left( \left( \Im f \right)^2 + \sqrt{\left( \Im f \right)^4 + \left( \Im f - 2 \Im g\right)^2} \right)},
        \end{align*}
        and $\bold{B}\left( \Im f, \Im g \right)$ is bounded since $\left( x, y \right) \mapsto \bold{B}\left( x, y \right)$ is continuous on any compact set of $\mathbb{R}^2\backslash \left\{(0,0)\right\}$ (recall both $\Im f$ and $\Im g$ are bounded) and has a finite limit at $\left(0,0\right)$, valued $1$, indeed
        \begin{equation*}
            \bold{B}\left( x, y \right) = \frac{1}{\left( 1-\rho^2\right) \frac{x^2}{x-2y} + \sqrt{\left( 1-\rho^2\right)^2\frac{x^4}{\left( x-2y \right)^2}+1}} \underset{\underset{x-2y \neq 0}{\left(x,y\right) \to \left(0,0\right)}}{\longrightarrow} 1.
        \end{equation*}

        \item Case $\rho \neq 0$ and $\Im f - 2 \Im g \neq 0$, then
        \begin{align*}
            \bold{A}\left( \Im f, \Im g \right) & = \frac{\rho \Im f}{\sqrt{\frac{1}{2} \left( \left( 1 - \rho^2 \right)\left( \Im f \right)^2 + \sqrt{\left( 1 - \rho^2 \right)^2\left( \Im f \right)^4 + \left( \Im f - 2 \Im g\right)^2} \right)}},\\
            \bold{B}\left( \Im f, \Im g \right) & = \frac{\frac{1}{2}\Im f - \Im g}{\frac{1}{2} \left( \left( 1 - \rho^2 \right)\left( \Im f \right)^2 + \sqrt{\left( 1 - \rho^2 \right)^2\left( \Im f \right)^4 + \left( \Im f - 2 \Im g\right)^2} \right)},
        \end{align*}
        are both bounded by continuity of $\left( x, y \right) \mapsto \bold{A}\left( x, y \right)$ and $\left( x, y \right) \mapsto \bold{B}\left( x, y \right)$ 
        on any compact set of $\mathbb{R}^2\backslash \left\{(0,0)\right\}$ (recall both $\Im f$ and $\Im g$ are bounded) and both functions have a finite limit at $\left(0,0\right)$, valued $0$ and $1$ respectively, obtained with similar arguments as in the previous case.
    \end{itemize}
\end{proof}
}

\begin{definition} \label{D:definitions_ig_levy_nig_ig}
(Inverse Gaussian, Lévy and NIG-IG distributions)
    \begin{itemize}
        \item We say X follows a Normal Inverse distribution, denoted $X \hookrightarrow \textit{IG}\left( \mu, \lambda \right)$ if its probability density writes
    \begin{equation*}
        f(x) = \sqrt{\frac{\lambda}{2\pi x^3}} \exp \left( -\frac{\lambda \left( x-\mu \right)^2}{2 \mu^2 x}\right),
    \end{equation*}
    where $\mu \in \mathbb{R}$, $\lambda>0$, or equivalently if the following equality holds true
    \begin{equation*}
        \mathbb{E} \left[ \exp \left( w X \right) \right] = \exp \left( \frac{\lambda}{\mu} \left( 1 - \sqrt{1-\frac{2\mu^2 w}{\lambda}} \right) \right), \quad w \in \mathbb{C}, \quad \Re w \leq 0.
    \end{equation*}

    \item We say $\tau$ follows a Lévy distribution, denoted $\tau \hookrightarrow \text{Lévy} \left( \mu, c\right)$, if its probability density writes
    \begin{equation*}
        \sqrt{\frac{c}{2\pi}}\frac{e^{-\frac{c}{2(x-\mu)}}}{\left( x-\mu \right)^{3/2}},
    \end{equation*}
    where $\mu \in \mathbb{R}$, $c>0$, or equivalently if the following equality holds true
    \begin{equation*}
        \mathbb E\left[\exp \left( w \tau \right)\right] = \exp \left( \mu w - \sqrt{-2c w} \right), \quad w \in \mathbb{C}, \quad \Re w \leq 0.
    \end{equation*}

    \item We say $\left( X, Y \right)$ follows a Normal Inverse Gaussian - Inverse Gaussian distribution, denoted $\left( X, Y \right) \hookrightarrow \text{NIG-IG} \left( \alpha, \beta, \mu, \delta, \lambda \right)$, if its characteristic function writes
    \begin{equation*}
        \mathbb E \left[ \exp \left( iu X + iv Y \right) \right] = \exp \left[ i \mu u + \delta \left( \sqrt{\alpha^2 - \beta^2} -\sqrt{\alpha^2 - 2i\lambda v -\left( \beta + iu \right)^2} \right) \right], \quad u, v \in \mathbb{R},
    \end{equation*}
    \end{itemize}
\end{definition}

\section{Additional plots} \label{ss:additional_plots}

\begin{figure}[H]
\begin{center}
\includegraphics[scale=0.3,angle=0]{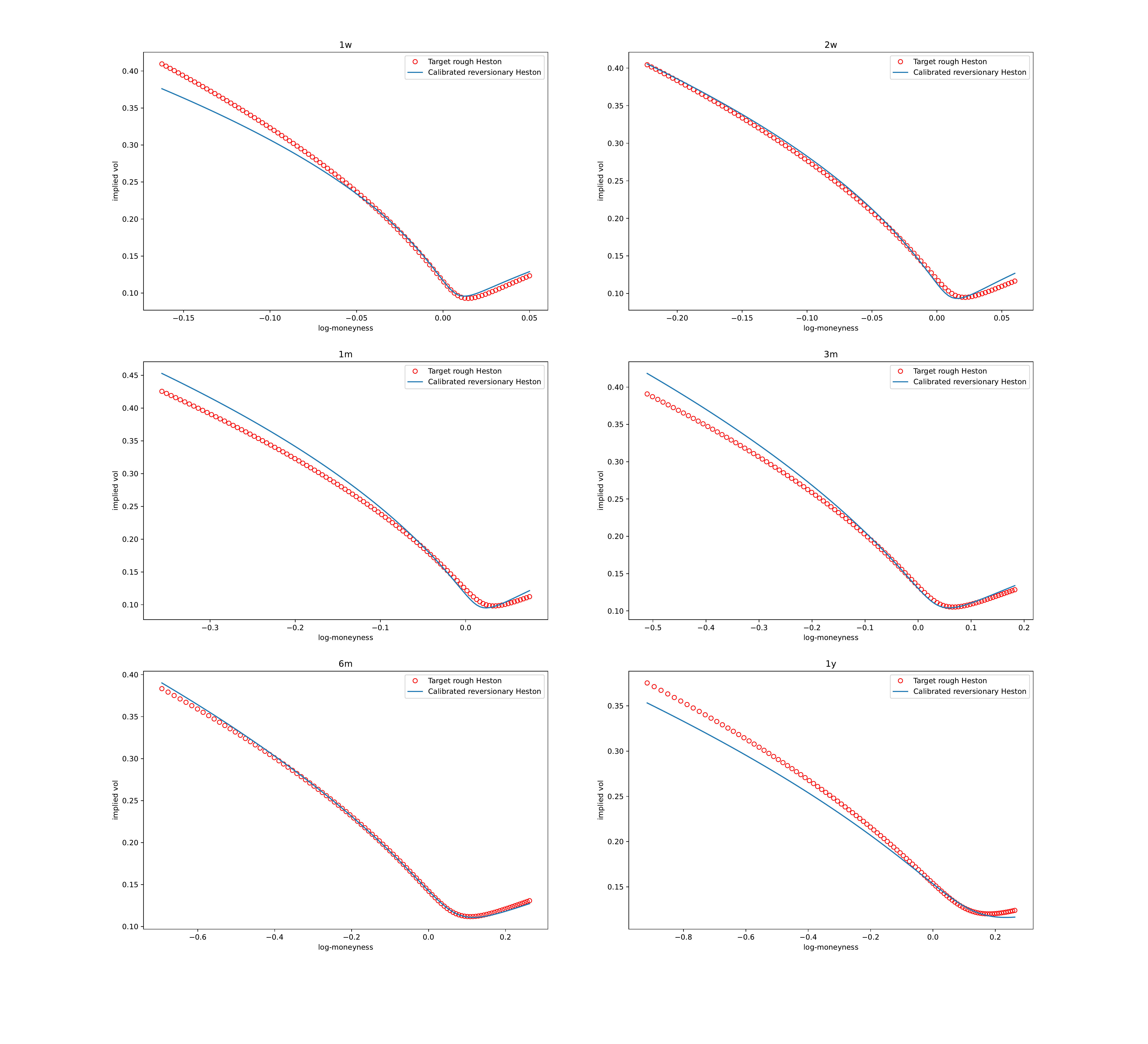}
\vspace*{-0.2in}
\caption{Smiles comparison between target rough Heston with parameters \eqref{eq:roughHestonParams}, with $H=0$, and reversionary Heston with calibrated parameters from the second row of Table \eqref{tab:calibrated_H_eps} for different maturities from one week to one year.}
\label{fig:smilesRoughHestonH0RevHeston}
\end{center}
\end{figure}

\begin{figure}[H]
\begin{center}
\includegraphics[scale=0.3,angle=0]{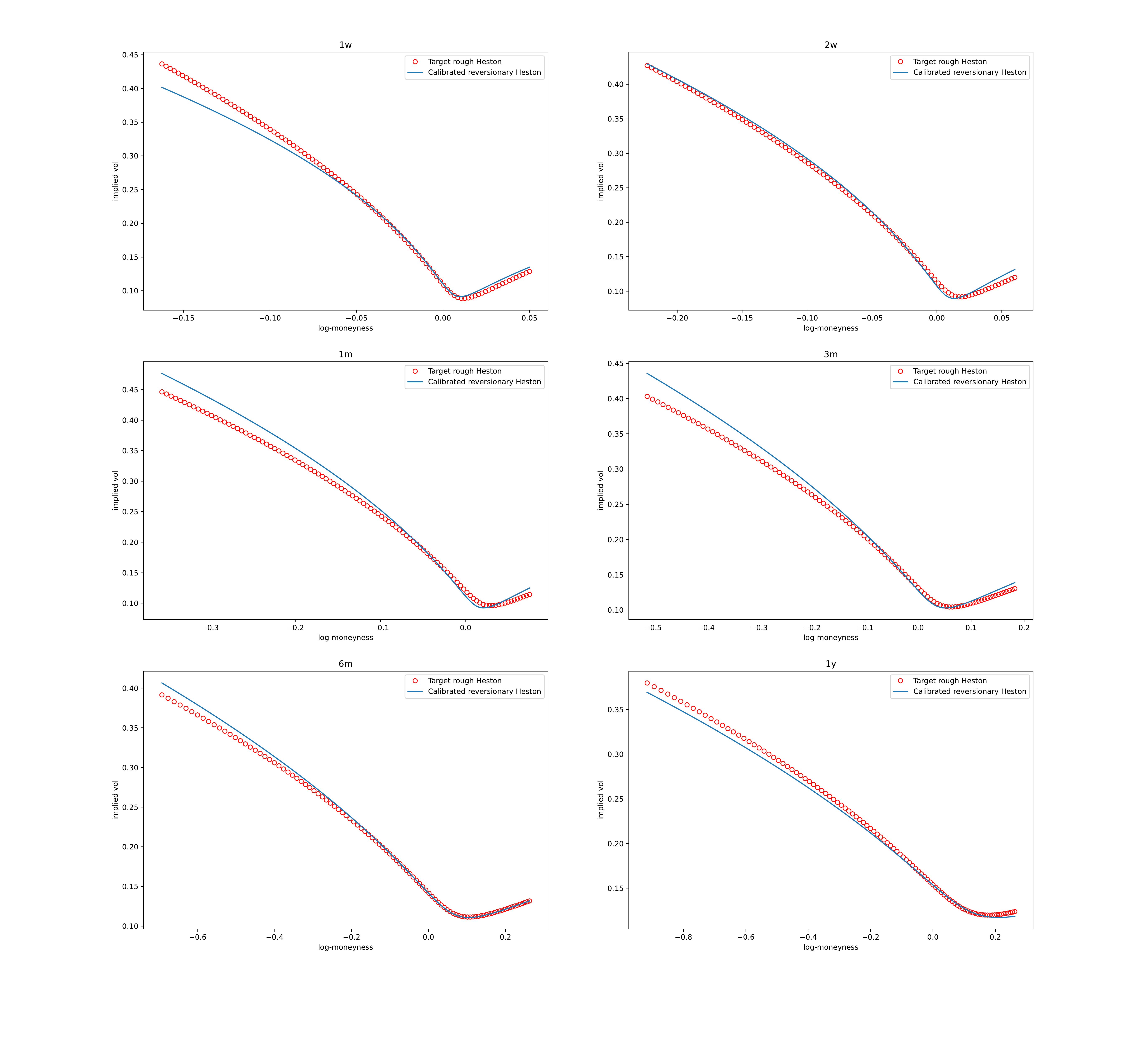}
\vspace*{-0.2in}
\caption{Smiles comparison between target rough Heston with parameters \eqref{eq:roughHestonParams}, with $H=-0.05$, and reversionary Heston with calibrated parameters from the third row of Table \eqref{tab:calibrated_H_eps} for different maturities from one week to one year.}
\label{fig:smilesRoughHestonHminus0dot05RevHeston}
\end{center}
\end{figure}

\bibliographystyle{plainnat}
\bibliography{bibl.bib}

\end{document}